\documentclass[aps,pra,nofootinbib,twocolumn,superscriptaddress]{revtex4-1}

\usepackage{epsfig}
\usepackage{graphicx}
\usepackage{amssymb}
\usepackage{amsmath}
\usepackage{amsthm}
\usepackage{colonequals}
\usepackage{hyperref}
\usepackage{epstopdf}
\usepackage{color}

\newcommand{\nn}{\nonumber \\}
\newcommand{\bra}[1]{\langle{#1}|}
\newcommand{\ket}[1]{|{#1}\rangle}
\newcommand{\braket}[2]{\langle{#1}|{#2}\rangle}

\newcommand{\Sj}{\sigma_j}
\newcommand{\UN}{N}
\newcommand{\HN}{M}
\newcommand{\hj}{\ell}
\newcommand{\hm}{L}
\newcommand{\sparseness}{D}
\newcommand{\s}{s}
\newcommand{\short}{\aleph}
\newcommand{\F}{F}

\newcommand{\La}{\Lambda}
\newcommand{\Li}{\Lambda_1}
\newcommand{\Lm}{\Lambda_{\max}}
\newcommand{\Ha}{\|H\|}
\newcommand{\Hi}{\|H\|_1}
\newcommand{\Hm}{\|H\|_{\max}}
\newcommand{\sca}{\Upsilon}

\newcommand{\phia}{\varphi}
\newcommand{\phib}{\phi}
\newcommand{\psia}{\eta}
\newcommand{\psib}{\psi}
\newcommand{\psic}{\zeta}

\newcommand{\CC}{\mathbb{C}}
\newcommand{\defeq}{\colonequals}
\newcommand{\poly}{\mathop{\mathrm{poly}}}
\newcommand{\diag}{\mathop{\mathrm{diag}}}

\newcommand{\unlaziness}{\varepsilon}
\newcommand{\laz}{\epsilon}
\newcommand{\error}{\delta}
\newcommand{\splitsmall}{\xi}
\newcommand{\numprecision}{\zeta}

\newcommand{\brk}{{\rm break}}

\newcommand{\ceil}[1]{\lceil{#1}\rceil}

\newtheorem{theorem}{Theorem}
\newtheorem{lemma}[theorem]{Lemma}
\newtheorem{corollary}[theorem]{Corollary}

\begin{document}

\title{Black-box Hamiltonian simulation and unitary implementation}

\author{Dominic W. Berry}
\affiliation{Institute for Quantum Computing, University of Waterloo, Waterloo, Ontario N2L 3G1, Canada}
\affiliation{Department of Physics and Astronomy, Macquarie University, Sydney, NSW 2109, Australia}
\author{Andrew M. Childs}
\affiliation{Institute for Quantum Computing, University of Waterloo, Waterloo, Ontario N2L 3G1, Canada}
\affiliation{Department of Combinatorics \& Optimization, University of Waterloo, Waterloo, Ontario N2L 3G1, Canada}

\begin{abstract}
We present general methods for simulating black-box Hamiltonians using quantum walks. These techniques have two main applications: simulating sparse Hamiltonians and implementing black-box unitary operations.
In particular, we give the best known simulation of sparse Hamiltonians with constant precision. Our method has complexity linear in both the sparseness $\sparseness$ (the maximum number of nonzero elements in a column) and the evolution time $t$, whereas previous methods had complexity scaling as $\sparseness^4$ and were superlinear in $t$.
We also consider the task of implementing an arbitrary unitary operation given a black-box description of its matrix elements.  Whereas standard methods for performing an explicitly specified $\UN\times\UN$ unitary operation use $\tilde O(\UN^2)$ elementary gates, we show that a black-box unitary can be performed with bounded error using $O(\UN^{2/3} (\log\log \UN)^{4/3})$ queries to its matrix elements. In fact, except for pathological cases, it appears that most unitaries can be performed with only $\tilde O(\sqrt{\UN})$ queries, which is optimal.
\end{abstract}
\maketitle

%%%%%%%%%%%%%%%%%%%%%%%%%%%%%%%%%%%%%%%%%%%%%%%%%%%%%%%%%%%%%%%%%%%%%%%%%%%%%%%%
\section{Introduction}

One of the major applications of quantum computation is the simulation of Hamiltonian dynamics.  Hamiltonian simulation is the basis for simulating quantum systems---the original motivation for quantum computers \cite{Feynman82}---and also has applications to quantum algorithms \cite{Childs03,Farhi08,Childs09b,Harrow09}.

An explicit procedure for simulating local Hamiltonians on a quantum computer was given by Lloyd \cite{Lloyd96}. This result was later substantially generalized to the simulation of sparse Hamiltonians by Aharonov and Ta-Shma \cite{Aharonov03}. References \cite{Childs04,Berry07} improved these results by providing a simulation scheme with complexity that scales close to linearly in the evolution time $t$ and as the fourth power of the sparseness parameter $\sparseness$, the maximum number of nonzero elements in a column.

In this paper, we consider the general task of simulating Hamiltonians, without necessarily assuming sparsity.  As particular applications, we provide improved methods for simulating sparse Hamiltonians and implementing unitary transformations. Similar to previous work on Hamiltonian simulation \cite{Aharonov03,Childs04,Berry07}, we assume throughout that the Hamiltonian is specified by an oracle.
Specifically, in our model,
a black-box function computes the matrix element $H_{jk}$ for any desired row index $j \in \{1,2,\ldots,\HN\}$ and column index $k \in \{1,2,\ldots,\HN\}$.
To simultaneously treat the case of sparse Hamiltonians, we also consider a black box that computes the positions of the nonzero matrix elements.  (See Sec.\ \ref{sec:model} for a detailed discussion of our model and its relationship to prior work.)

An algorithm for computing the matrix elements of $H$ can be used to construct such a black box. In particular, this means that the black-box model applies to common physical Hamiltonians such as those considered by Lloyd \cite{Lloyd96}. Such Hamiltonians are a sum of local terms, each acting on a limited number of subsystems. Thus they are sparse, and there is an efficient method for computing the matrix elements of the overall Hamiltonian. An advantage of the black-box model of Hamiltonian simulation is that it can be applied not only to physical Hamiltonians, but as a basis for designing algorithms for other problems \cite{Childs03,Childs09b,Harrow09}.

In contrast to previous work on Hamiltonian simulation \cite{Lloyd96,Aharonov03,Childs03,Childs04,Berry07}, most of which is based on Lie-Trotter-Suzuki formulae, we use a new approach based on a quantum walk \cite{Childs08}. A limitation of the Lie-Trotter-Suzuki approach is that it relies on limiting the error by using many short time steps. As a result, the complexity of the simulation always scales superlinearly in $t$. The quantum walk approach provides scaling that is strictly linear in $t$ \cite{Childs08}, which is known to be optimal \cite{Berry07}. Another limitation of the Lie-Trotter-Suzuki approach is that it relies on decomposing the Hamiltonian into 1-sparse matrices, which results in poor scaling in the sparseness $\sparseness$. In Ref.\ \cite{Berry07} the scaling was $\tilde O(\sparseness^4)$, which was recently improved to $\tilde O(\sparseness^3)$ \cite{Childs11}.\footnote{We use a tilde to indicate that subpolynomial scaling is ignored---that is, $f=\tilde O(g)$ if $f=O(g^{1+\eta})$ for any $\eta>0$.}
In contrast, by using quantum walks we improve the scaling to linear in $\sparseness$.  This represents the best known constant-precision simulation of sparse Hamiltonians.  Furthermore, if we are willing to accept superlinear scaling in $t$, the scaling in $\sparseness$ may be improved to $\tilde O(\sparseness^{2/3})$ in general, and in many cases to $\tilde O(\sparseness^{1/2})$.

The quantum walk approach to Hamiltonian simulation was proposed in Ref.\ \cite{Childs08}, though without providing an explicit method to implement the steps of the quantum walk in the general case. Here we present a complete method for Hamiltonian simulation by showing how to implement the steps of the quantum walk for a general Hamiltonian that may or may not be sparse.
We use the method of Ref.\ \cite{Childs08} together with a range of other tools, which we combine and improve on in nontrivial ways to obtain our final result.
In particular, our approach introduces the following techniques.
\begin{enumerate}
\item In Sec.\ \ref{sec:spar}, we modify the method of Ref.\ \cite{Childs08} by using phase estimation to correct a lazy quantum walk, giving a more efficient simulation.
\item In Sec.\ \ref{sec:impl}, we describe how to perform steps of the quantum walk using state preparation by amplitude amplification (similar to a method proposed in Ref.\ \cite{Grover00}), improving the efficiency in the non-sparse case. 
\item We modify the state to be prepared by using an ancilla qubit to satisfy an orthogonality condition for the lazy quantum walk (compare Eqs.\ \eqref{eq:alt2} and \eqref{eq:alt3a}). This facilitates more efficient state preparation by amplitude amplification in Sec.\ \ref{sec:impl} (outperforming direct application of Ref.\ \cite{Grover00}), and even allows us to prepare the state in only $O(1)$ queries in the sparse case described in Sec.\ \ref{sec:spar}.
\item In Sec.\ \ref{sec:break}, we further improve the simulation in the non-sparse case by decomposing the Hamiltonian as a sum of terms and recombining these terms using Lie-Trotter-Suzuki formulae.  The decomposition depends on the magnitudes of the matrix elements, giving an approach that is fundamentally different from previous applications of Lie-Trotter-Suzuki formulae.
\end{enumerate}

While our results on simulating non-sparse Hamiltonians may be of interest in their own right, additional motivation for studying this problem comes from the related task of implementing general unitary transformations.  Standard methods for implementing an arbitrary $\UN \times \UN$ unitary transformation on a quantum computer work by decomposing it into a product of two-level unitary matrices \cite{Reck94,Barenco95} and performing the two-level unitaries via the Solovay-Kitaev theorem \cite{SoloKit1,SoloKit2,Harrow02}. This method uses $\UN^2 \poly(\log \UN)$ gates. Since counting arguments show that $\Omega(\UN^2)$ elementary gates are required even to approximate a general unitary transformation \cite{Barenco95,Knill95,Harrow02}, such an implementation is nearly optimal. 

Instead of considering an explicit unitary operation, we study the problem of performing a unitary transformation specified by a black box for its matrix elements, similar to the black box for a non-sparse Hamiltonian described above.  In such an oracle model, counting arguments for unitary implementation no longer apply, since the black box depends on the unitary.  While the question of how many queries are required to implement a general unitary in this model seems quite natural, to the best of our knowledge it has not been studied previously.

Implementation of unitaries is closely connected to Hamiltonian simulation, because one can implement a unitary by simulating a related Hamiltonian. Reference \cite{Jordan09} used this idea to provide a method for implementing sparse unitaries.  However, their approach relies on a decomposition into $1$-sparse Hamiltonians, so it performs poorly in the non-sparse case.

Using Hamiltonian simulation via quantum walks, we show that the complexity of implementing a general (non-sparse) $\UN \times \UN$ unitary scales with $\UN$ as $\tilde O(\UN^{2/3})$. We also present numerical evidence that typical unitaries can be implemented in only $\tilde O(\sqrt{N})$ queries (although it is possible to construct unitaries for which our method uses more queries). This is much less than the $\Omega(\UN^2)$ elementary gates required to implement a general unitary given an explicit description instead of a black box. The best lower bound we are aware of is $\Omega(\sqrt{\UN})$, because implementation of a black-box unitary operation can be used to solve a search problem \cite{Grovopt}.

Implementation of black-box unitary transformations is closely related to the task of preparing an $N$-dimensional quantum state given a black box for its amplitudes in a fixed basis.  Grover showed that the query complexity of this task is $\Theta(\sqrt{N})$ \cite{Grover00}. As mentioned above, we build on his technique in order to implement black-box unitaries. It is an open question whether black-box unitaries can be implemented in $O(\sqrt{N})$ queries in general, or if there is a fundamental separation between the query complexity of implementing unitaries and the query complexity of preparing states.

Since implementing unitary transformations is a basic task in quantum computation, we expect this result to have applications to quantum algorithms.  For example, our approach could serve as an alternative to previous methods for efficiently implementing general unitary transformations on a logarithmic number of qubits, with improved performance provided the matrix elements can be computed quickly.

The remainder of this article is organi{\s}ed as follows.
In Sec.~\ref{sec:summ} we give a technical summary of our main contributions.
Then, in Sec.~\ref{sec:hamsim}, we summari{\s}e the method of Ref.~\cite{Childs08} for simulating Hamiltonian evolution.
Our main result, that a black-box Hamiltonian can be simulated with $\tilde O(\sparseness^{2/3})$ queries to its matrix elements, is proven in Sec.~\ref{sec:break}, building on a foundation established in Secs.~\ref{sec:spar} through \ref{sec:full}.
Some of these intermediate results may be of interest in their own right; in particular, in Sec.\ \ref{sec:spar}, we present a simple method to perform the steps of the quantum walk of Ref.~\cite{Childs08} that is especially suitable for sparse Hamiltonians.
We explain how Hamiltonian simulation can be used to implement black-box unitaries in Sec.~\ref{sec:unit}.
In Sec.~\ref{sec:examp} we give some examples of the simulation as applied to particular unitary operations.
We conclude in Sec.~\ref{sec:conc} with a summary of the results and a discussion of some open problems.

\section{Model and results}
\label{sec:summ}

\subsection{Model}
\label{sec:model}

We formulate the Hamiltonian simulation and unitary implementation problems using an oracle model, in which a description of the Hamiltonian or unitary is provided by a black box.
For Hamiltonian simulation, the matrix elements of some Hermitian matrix $H \in \CC^{\HN \times \HN}$ are given by a black box $O_H$ acting as
\begin{equation}
\label{eq:horacle}
O_H \ket{j,k}\ket{z} = \ket{j,k}\ket{z \oplus H_{jk}},
\end{equation}
where $j,k \in \{1,2,\ldots,\HN\}$.
Here the matrix element $H_{jk}$ is represented by its real and imaginary parts written in binary, and $\oplus$ denotes the bitwise XOR of such representations.
Similarly, for the problem of implementing a unitary $U \in \CC^{\UN \times \UN}$, we are given a black box $O_U$ acting as
\begin{equation}
\label{eq:uoracle}
O_U \ket{j,k}\ket{z} = \ket{j,k}\ket{z \oplus U_{jk}},
\end{equation}
where $j,k \in \{1,2,\ldots,\UN\}$.
The error of the simulation or implementation must be no greater than $\error$ as quantified by the trace distance.
In practice, the black box $O_H$ or $O_U$ provides the matrix elements to some finite precision, $\numprecision$, using $O(\log\frac{1}{\numprecision})$ qubits. We assume that $\numprecision\ll\error$, so the imprecision in this approximation does not affect the analysis.

We can also take advantage of sparsity if it is possible to compute the positions of nonzero matrix elements.
Specifically, suppose there are at most $\sparseness$ nonzero elements in each row or column. For Hamiltonian simulation, suppose that in addition to the black box $O_H$, we are given a black box $O_\F$ acting as
\begin{equation}
O_\F \ket{j,k} = \ket{j,f(j,k)}
\end{equation}
for any $j \in \{1,2,\ldots,\HN\}$ and $k \in \{1,2,\ldots,\sparseness\}$,
where the function $f(j,k)$ gives the row index of the $k$th nonzero element in column $j$ (or the row index of any zero element when there are fewer than $k$ nonzero elements in column $j$).

Note that $O_\F$ computes the row index in place.  In contrast, some previous work on simulating sparse Hamiltonians \cite{Berry07,Childs11} assumes that a single query only computes $f(j,k)$ given $j$ and $k$, i.e., performs the isometry $O'_\F$ acting as $O'_\F \ket{j,k} = \ket{j,k,f(j,k)}$. The oracle $O_\F$ can be used to produce such an oracle in one query, simply by copying $k$ to a third register before calling $O_\F$. Therefore, all the upper bounds for the complexity in \cite{Berry07,Childs11} hold for the oracle $O_\F$. Furthermore, the algorithms in \cite{Berry07,Childs11} do not depend on this aspect of the oracle, so changing the oracle in those algorithms would not lead to improved upper bounds.
Thus our results may be directly compared with those of \cite{Berry07,Childs11}.

To construct the oracle $O_F$, it suffices to first compute $f(j,k)$ for a given $j$, and then compute $k$ given $j$ and $f(j,k)$ in order to erase the register encoding $k$.  In contrast, $O'_\F$ does not uncompute $k$.  For realistic cases such as the local Hamiltonians considered by Lloyd \cite{Lloyd96}, determining $k$ given $j$ and $f(j,k)$ is not difficult, so $O_\F$ is a realistic representation of the resources used.

If desired, one can quantify the resources used by our simulations in terms of queries to the black box $O'_\F$ instead of to $O_\F$.  Even if $k$ cannot be computed directly, one can implement $O_\F$ with $O'_\F$ using additional queries to uncompute $k$. If the function $f$ provides the nonzero elements in sorted order, one can find $k$ by binary search, increasing the number of queries by a factor of only $\log\sparseness$. In general, one can use Grover's algorithm to find $k$, increasing the number of queries by a factor of $O(\sqrt\sparseness)$.

Another model used by some papers on sparse Hamiltonian simulation is that for any given $j$, a single query reveals all the nonzero entries in the $j$th column, i.e., the values $f(j,1),\ldots,f(j,\sparseness)$ \cite{Aharonov03,Childs03,Childs04}. With that model, the black box $O_\F$ can be implemented using two queries, whereas $\sparseness$ calls to $O_\F$ are required to compute all the values $f(j,1),\ldots,f(j,\sparseness)$.  As Refs.\ \cite{Aharonov03,Childs03,Childs04} are primarily concerned with showing polynomial scaling in $\sparseness$, such a difference is unimportant.

We emphasi{\s}e that in the present work, we do not require $\sparseness=\poly(\log \HN)$; our methods apply for any $\sparseness\le\HN$. If the Hamiltonian is not sparse, or if it is sparse but the nonzero elements are in unknown positions, then we can simply take $\sparseness=\HN$ and let $O_\F$ be the identity operation. Thus, all the results of the paper hold for non-sparse cases, with $\sparseness=\HN$.  In particular, when considering the problem of unitary implementation, we do not assume sparsity, so the black box $O_\F$ is not required.

We assume that information about the Hamiltonian or unitary can only be obtained by querying the oracle, so in particular we do not know the norms of $H$ or $U$ (except for the trivial fact that $\|U\|=1$). However, we assume that we do have upper bounds on various norms: we are given constants $\La,\Li,\Lm$ satisfying $\La\ge\Ha$, $\Li\ge\Hi$, and $\Lm\ge\Hm$, where $\Ha$ denotes the spectral norm of $H$, $\Hi \defeq \max_j \sum_{k=1}^\HN |H_{jk}|$, and $\Hm \defeq \max_{j,k} |H_{jk}|$.

\subsection{Results}

Our first main result, proved in Sec.\ \ref{sec:spar}, is that a Hamiltonian can be simulated with scaling linear in both $\Ha t$ and $\sparseness$.

\begin{theorem}
\label{theorem1}
For a given Hamiltonian $H$, let $\La\ge\Ha$ and $\Lm\ge\Hm$.  Then the evolution under $H$ for time $t$ can be simulated with error at most $\error\in(0,1]$ using
\begin{equation}
\label{eq:other}
O\left( \frac{\La t}{\sqrt\error} + \sparseness \Lm t + 1 \right)
\end{equation}
queries to $O_H$ and $O_\F$.
\end{theorem}

This result is suitable for simulation of sparse Hamiltonians. The next main result, proved in Sec.\ \ref{sec:largenorms}, gives improved scaling in $\sparseness$ at the expense of worse scaling in $\Ha t$.  This result may be preferable for non-sparse Hamiltonians.

\begin{theorem}
\label{theorem2}
Let $\La\ge\Ha$.  The evolution under the Hamiltonian $H$ for time $t$ can be simulated with error at most $\error\in(0,1]$ using
\begin{equation}
\label{eq:th2}
O\left( \sparseness^{2/3}[(\log\log\sparseness)\La t]^{4/3} \error^{-1/3} \right)
\end{equation}
queries to $O_H$ and $O_\F$, provided $\error\sparseness > \La t > \sqrt\error$.
\end{theorem}

Using the correspondence between Hamiltonian simulation and unitary implementation described in Sec.\ \ref{sec:unit}, this easily implies our main result on the implementation of black-box unitaries (see Sec.\ \ref{sec:unit} for the proof).

\begin{corollary}
\label{corollary1}
A black-box unitary operation $U$ can be implemented with error at most $\error\in(0,1]$ using
\begin{equation}
\label{eq:unsca}
O\left( \UN^{2/3} (\log\log\UN)^{4/3} \error^{-1/3} \right).
\end{equation}
queries to $O_U$.
\end{corollary}

Although the above results are the best we are able to show for general non-sparse Hamiltonians and unitaries, we believe that our methods are typically more efficient.  Theorem \ref{theorem2} and Corollary \ref{corollary1} are based on a decomposition of the Hamiltonian into a sum of terms, where the nonzero matrix elements of each term have comparable size.  In the worst-case analysis of Sec.\ \ref{sec:largenorms}, we must take into account the possibility that in this decomposition, the spectral norms of the individual terms could be much larger than the spectral norm of the total Hamiltonian. Numerically, we find that this does not occur when selecting matrices at random, only when matrices are specifically designed to cause this behaviour (see Sec.\ \ref{sec:norms}). Assuming that all terms in the decomposition have comparable norms, we show in Sec.\ \ref{sec:smallnorms} that the number of queries to simulate a Hamiltonian with $\Ha \le \La$ is
\begin{equation}
O\left( (\La t)^{3/2} \sqrt{\sparseness/\error} (\log\sparseness)^{7/4} \right);
\end{equation}
correspondingly, a black-box unitary can be implemented using
\begin{equation}
O\left( \sqrt{\UN/\error} (\log\UN)^{7/4} \right)
\end{equation}
queries.

\section{Review of Hamiltonian simulation}
\label{sec:hamsim}

In this section we summari{\s}e an approach to Hamiltonian simulation based on discrete-time quantum walks \cite{Childs08}.
Throughout, $\HN$ denotes the dimension of the Hilbert space that a black-box Hamiltonian acts on, and $\UN$ the dimension of the space that a black-box unitary transformation acts on.
To construct a discrete-time quantum walk from a given Hamiltonian $H$,
the Hilbert space is expanded from $\CC^{\HN}$ to $\CC^{\HN+1} \otimes \CC^{\HN+1}$. A step of the discrete-time quantum walk is described by a unitary operator
\begin{equation}
\label{eq:walkstep}
V \defeq iS(2TT^\dagger -\openone).
\end{equation}
Here the operator $S$ swaps the two registers, i.e., $S\ket{j,k}=\ket{k,j}$ for all $j,k \in \{1,2,\ldots,\HN+1\}$. The operator $T$ is the isometry
\begin{equation}
\label{eq:tdef}
T \defeq \sum_{j=1}^\HN \ket{\psia_j}\bra{j}
\end{equation}
mapping $\ket{j}$ to $\ket{\psia_j} \defeq \ket{j}\ket{\phia_j}$, where
\begin{equation}
\label{eq:alt2}
\ket{\phia_j} \defeq \sqrt{\frac{\laz}{\Hi}}\sum_{k=1}^{\HN}\sqrt{H_{jk}^*}\ket{k}+\sqrt{1-\frac{\laz\Sj}{\Hi}}\ket{\HN+1}
\end{equation}
with
\begin{equation}
\label{eq:sigmaj}
\Sj \defeq \sum_{k=1}^{\HN} |H_{jk}|
\end{equation}
(cf.\ Eq.\ (27) of Ref.\ \cite{Childs08}).
Here $\laz \in (0,1]$ is a parameter that can be made small to obtain a \emph{lazy quantum walk}, and $\Hi \defeq \max_j \sum_{k=1}^{\HN}|H_{jk}|$.

The state $\ket{\phia_j}$ is chosen so that $\bra{\psia_j}S\ket{\psia_k}$ is proportional to $H_{jk}$.  We have
\begin{align}
  \bra{\psia_j}S\ket{\psia_k}
  &= \braket{j}{\phia_k} \braket{\phia_j}{k} \nonumber \\
  &= \frac{\epsilon}{\Hi} \sqrt{H_{kj}^*} \left(\sqrt{H_{jk}^*}\right)^* .
\label{eq:squareroots}
\end{align}
Note that caution is needed when choosing the sign of the square root.  Provided $H_{jk}$ is not a negative real number, it suffices to take the principal square root of both $H_{kj}^*$ and $H_{jk}^*$ (i.e., if $z = r e^{i \theta}$ for some $r \ge 0$ and $\theta \in (-\pi,\pi)$, define $\sqrt{z} \defeq \sqrt{r} e^{i \theta/2}$, so that $\sqrt{z} (\sqrt{z^*})^* = z$).  This choice ensures that Eq.\ \eqref{eq:squareroots} gives $\epsilon H_{jk} / \Hi$.  However, if $H_{jk}$ is a negative real number, this choice does not suffice.  Instead, for $H_{jk} \in (-\infty,0)$ with $j\ne k$, we take
\begin{equation}
\sqrt{H_{jk}^*} = \mathrm{sign}(j-k) \, i \sqrt{|H_{jk}|}.
\end{equation}
By taking a different sign above and below the diagonal, we ensure that Eq.\ \eqref{eq:squareroots} is negative, as required.

The above prescription does not handle the case where the Hamiltonian has negative diagonal elements.  To ensure that the diagonal entries are nonnegative, we simply add a multiple of the identity: given an upper bound $\Lm$ on $\Hm$, we replace $H$ with $H+\Lm\openone$.  This only changes the Hamiltonian evolution for time $t$ by a global phase of $e^{-i \Lm t}$, and the relevant norms of $H$ are increased by at most a factor of $2$.

The eigenvalues and eigenvectors of $V$ are closely related to those of $H$ \cite{Szegedy}. If we define $\tilde H$ to be the operator with matrix elements 
\begin{equation}
\label{eq:normalizedh}
\tilde H_{jk} \defeq \bra{\psia_j}S\ket{\psia_k},
\end{equation}
then Eq.\ \eqref{eq:squareroots} gives
\begin{equation}
\label{eq:h0generaldef}
\tilde H = \frac{\laz H}{\Hi}.
\end{equation}
The operators $H$ and $\tilde H$ have common eigenstates $\ket\lambda$. The corresponding eigenvalues for $H$ and $\tilde H$ are denoted $\lambda$ and $\tilde\lambda$, and are related by
\begin{equation}
\label{eq:lamrel}
\tilde \lambda = \frac{\laz\lambda}{\Hi}.
\end{equation}
Each eigenstate $\ket\lambda$ corresponds to two eigenvectors of $V$,
\begin{align}
\ket{\mu_\pm^\lambda} &\defeq \frac{1-e^{\pm i\arcsin\tilde\lambda}S}{\sqrt{2(1-\tilde\lambda^2})} T\ket{\lambda},
\end{align}
with eigenvalues
\begin{align}
\label{eq:veigenvalues}
\mu_\pm^\lambda &\defeq \pm e^{\pm i\arcsin\tilde\lambda}.
\end{align}

Reference \cite{Childs08} describes simulations of $H$ based on the quantum walk $V$.  One approach is to use phase estimation to (coherently) determine the value of $\tilde\lambda$.
Introducing a phase of $\exp({-i\lambda t}) = \exp({-i\tilde \lambda t\Hi/\laz})$ for each eigenvector $\ket\lambda$ simulates evolution under $H$ for time $t$.
More specifically, a general initial state has the form
\begin{equation}
\ket{\psib} = \sum_{\lambda,k} \psib_{\lambda,k}\ket{\lambda,k},
\end{equation}
where the index $k$ accounts for degenerate eigenvalues of $H$. Applying the operation $T$ yields
\begin{align}
T\ket{\psib} &= \sum_{\lambda,k}\psib_{\lambda,k}T\ket{\lambda,k} \nn
& \begin{aligned}
= \sum_{\lambda,k}\frac{\psib_{\lambda,k}}{\sqrt{2(1-\tilde\lambda^2)}}[(1-\tilde\lambda e^{-i\arccos\tilde\lambda})&\ket{\mu_+^\lambda,k}\\[-12pt]
+(1-\tilde\lambda e^{i\arccos\tilde\lambda})&\ket{\mu_-^\lambda,k}],
\end{aligned}
\end{align}
where the $\ket{\mu_\pm^\lambda,k}$ are eigenvectors of $V$, with the index $k$ labeling an orthonormal basis for each eigenspace. Applying the correct phase factor for each eigenspace gives
\begin{equation}
T e^{-iHt}\ket{\psib} = \sum_{\lambda,k}e^{-i\lambda t}\psib_{\lambda,k}T\ket{\lambda,k}.
\end{equation}
Applying $T^\dagger$ then gives $e^{-iHt}\ket{\psib}$, the desired time-evolved state.

Phase estimation can provide an estimate of $\mu_\pm^\lambda$ with variance approximately $(\pi/d)^2$ using $d$ applications of $V$. The variance in $\lambda$ is then $O(\Hi^2/\laz^2 d^2)$, which translates to an error in the state of $O(\Hi t/\laz d)$. If the allowed error is $\error$, then the simulation can be achieved with $d=O(\Hi t/\error)$ by taking $\laz=1$. (Note that the $\error$ used here is the square root of that used in Ref.\ \cite{Childs08}, because that paper considered a lower bound on the fidelity of $1-\error$, whereas we take $\error$ to be an upper bound on the trace distance.)

\section{Sparse Hamiltonian simulation}
\label{sec:spar}

The simulation scheme presented in Ref.\ \cite{Childs08} quantifies the complexity in terms of the number of quantum walk steps. To simulate a black-box Hamiltonian, we require a method to perform these steps. In this section we describe a simple approach to this problem, thereby providing a Hamiltonian simulation method suitable for the sparse case.

The walk step is composed of two operators, the swap $S$ and a reflection $2TT^\dagger -\openone$. The operation $S$ is easy to implement; the difficulty lies in implementing the reflection. It is given explicitly by
\begin{equation}
\label{eq:tref}
2TT^\dagger -\openone = \sum_{j=1}^\HN \ket{j}\bra{j}\otimes(2\ket{\phia_j}\bra{\phia_j}-\openone).
\end{equation}
That is, it is a reflection about $\ket{\phia_j}$ conditional on the state $\ket{j}$ in the first register. To perform this reflection, it suffices to give a procedure for preparing $\ket{\phia_j}$ from the $\ket{0}$ state:  by performing inverse state preparation, reflecting about $\ket{0}$, and then performing state preparation, we effectively reflect about $\ket{\phia_j}$.

Black-box preparation of an $\HN$-dimensional quantum state can clearly be performed in $\HN$ queries, but this would introduce an overall multiplicative factor of $\HN$ in the complexity of the implementation. The overall query complexity of the simulation would then be $O(\HN\Hi t/\error)$. Taking advantage of sparsity reduces this to $O(\sparseness\Hi t/\error)$, but Theorem~\ref{theorem1} uses even fewer queries.

Our improved simulation uses the following insight. We modify the state $\ket{\phia_j}$ to
\begin{equation}
\label{eq:alt3a}
\ket{\phib_j} \defeq
\sqrt{\frac{\laz}{\Hi}}\sum_{k=1}^{\HN}\sqrt{H_{jk}^*}\ket{k}\ket{0}+\sqrt{1-\frac{\laz \Sj}{\Hi}}\ket{\psic_j}\ket{1},
\end{equation}
where $\ket{\psic_j}$ is some superposition of the $\ket{k}$.
That is, we append an ancilla qubit, and replace $\ket{M+1}$ with $\ket{\psic_j}\ket{1}$.
The second term of Eq.\ \eqref{eq:alt3a}, flagged by a $\ket{1}$ state in the ancilla qubit, takes the place of the $\ket{\HN+1}$ state in Eq.\ \eqref{eq:alt2}. Thus the discrete-time quantum walk takes place in $\CC^{2\HN} \otimes \CC^{2\HN}$, although it is effectively confined to a subspace of dimension $(\HN+1)^2$.

To take account of the fact that $\Hi$ may not be known exactly, we replace $\laz$ with $\unlaziness=\laz\Li/\Hi$, where $\Li$ is a known upper bound on $\Hi$. Then we can alternatively express the definition of $\ket{\phib_j}$ as
\begin{equation}
\label{eq:alt3}
\ket{\phib_j} \defeq \sqrt{\frac{\unlaziness}{\Li}}\sum_{k=1}^{\HN}\sqrt{H_{jk}^*}\ket{k}\ket{0}+\sqrt{1-\frac{\unlaziness \Sj}{\Li}}\ket{\psic_j}\ket{1}.
\end{equation}
Note that the restriction $\laz\le 1$ implies that $\unlaziness\le \Li/\Hi$. The relation between the eigenvalues $\lambda$ and $\tilde\lambda$ can be expressed in terms of $\unlaziness$ as
\begin{equation}
\lambda = \tilde\lambda \Li/\unlaziness.
\end{equation}

Provided $\unlaziness$ is sufficiently small, we can prepare this state using a constant number of queries. In particular:
\begin{lemma}
\label{lemma1}
The state $\ket{\phib_j}$ in Eq.\ \eqref{eq:alt3} can be prepared in $O(1)$ calls to the oracles $O_H$ and $O_\F$ provided
$\unlaziness \in (0,\Li/\sparseness\Lm]$, 
where $\Lm\ge\Hm$ and $\Li\ge\Hi$.
\end{lemma}
\begin{proof}
First, prepare an equal superposition over $\ket 1$ to $\ket\sparseness$ in the first register and initiali{\s}e the ancilla qubit to $\ket{0}$, giving
\begin{equation}
\frac{1}{\sqrt{\sparseness}} \sum_{k=1}^{\sparseness} \ket{k}\ket{0}.
\end{equation}
Querying the black box $O_\F$ changes this to
\begin{equation}
\ket{\phib_j^a} \defeq \frac{1}{\sqrt{\sparseness}} \sum_{k\in \F_j} \ket{k}\ket{0},
\end{equation}
where $\F_j$ is the set of indices given by $O_\F$ on input $j$.

Next we transform $\ket{\phib_j^a}$ to
\begin{align}
\label{eq:xstate}
\frac{1}{\sqrt{\sparseness}}\sum_{k\in \F_j}
\ket{k}\left[\sqrt{\frac{H_{jk}^*}{X}}\ket{0}  +\sqrt{1-\frac{|H_{jk}|}{X}}\ket{1}\right]
\end{align}
where $X = \Li/\unlaziness\sparseness$. The most important requirement on $X$ is that $X\ge \Lm$, so $X\ge \Hm$ and the amplitude for the $\ket{0}$ state has magnitude at most $1$. Because of the requirement that $\unlaziness \le \Li/\sparseness\Lm$, taking $X = \Li/\unlaziness\sparseness$ ensures that $X\ge \Lm$. In addition, for this value of $X$, Eq.\ \eqref{eq:xstate} has the form of $\ket{\phib_j}$ for some choice of $\ket{\psic_j}$.

The state $\ket{\phib_j^a}$ can be transformed to Eq.\ \eqref{eq:xstate} by computing $H_{jk}$ in an ancilla register with the black box $O_H$ and using this value to perform a controlled rotation on the qubit.  Applying $O_H$ again uncomputes the ancilla storing $H_{jk}$.  Note that the rotation can be performed with error at most $\error$ using $\poly(\log\frac{1}{\error})$ operations \cite{SoloKit1,SoloKit2}, but this factor is not included in the analysis since we focus on the query complexity.
\end{proof}

Next, we improve the Hamiltonian simulation method reviewed in Sec.\ \ref{sec:hamsim} by combining a lazy quantum walk with phase estimation.

\begin{lemma}
\label{lemma2}
Let $\Li t/\unlaziness \in \mathbb{Z}$ and $\unlaziness\in(0,1]$.
Evolution under $H$ for time $t$ can be simulated using $O(\Li t/\unlaziness)$ steps of the quantum walk defined by the states $\ket{\phib_j}$ from Eq.\ \eqref{eq:alt3}, with error $O(\La^2\unlaziness^2/\Li^2)$.
\end{lemma}

\begin{proof}
Given an arbitrary input state, first we (coherently) determine the sign $\pm$ of $\mu_\pm^\lambda$ (recall Eq.~\eqref{eq:veigenvalues}). The phase of $\mu_+^\lambda$ is $\arcsin\tilde\lambda$, whereas the phase of $\mu_-^\lambda$ is $\pi-\arcsin\tilde\lambda$. Performing phase estimation with one bit of precision on $V$ gives probabilities of measuring $+$ or $-$, given that the eigenvalue is $\mu_+^\lambda$ or $\mu_-^\lambda$, of
\begin{align}
\Pr(+|\pm)&=\frac{1\pm \sqrt{1-\tilde\lambda^2}}{2}, \\
\Pr(-|\pm)&=\frac{1\mp \sqrt{1-\tilde\lambda^2}}{2}.
\end{align}
The probability of error is therefore $O(\tilde\lambda^2)$. Since $\tilde\lambda=\unlaziness\lambda/\Li$ and $\lambda \le \La$, the error due to misidentification of the sign is $O(\La^2\unlaziness^2/\Li^2)$.

Having estimated the sign, we can apply a lazy quantum walk more accurately than in Ref.\ \cite{Childs08}. If the sign of $\mu_\pm^\lambda$ is $-$, we apply $V^d$, where
$d = \Li t/\unlaziness$ is an integer (chosen so that $\lambda t = \tilde\lambda d$), giving a phase factor of
\begin{equation}
(\mu_-^\lambda)^d
= (-e^{-i\arcsin \tilde\lambda})^d
= (-1)^d e^{-i\lambda t} + O(d \tilde\lambda^3).
\end{equation}
(The sign can be corrected if $d$ is odd.)
Similarly, if the sign of $\mu_\pm^\lambda$ is $+$, we apply $(V^\dagger)^d$, giving a phase factor
\begin{equation}
[(\mu_+^\lambda)^*]^d 
= (e^{-i\arcsin \tilde\lambda})^d
= e^{-i\lambda t} + O(d \tilde\lambda^3).
\end{equation}
In either case, the error in the final state is
\begin{equation}
O(d \tilde\lambda^3) =
O(\lambda^3 t (\unlaziness/\Li)^2) \le O(\La^3 t (\unlaziness/\Li)^2),
\end{equation}
considerably less than that for the method given in Ref.\ \cite[Theorem 2]{Childs08} for small $\unlaziness$.

We further reduce the error by using an estimate of $\tilde\lambda$ to correct the lazy quantum walk. After the lazy quantum walk implements the phase factor $e^{-id\arcsin\tilde\lambda}$, the estimate of $\tilde\lambda$ is used to correct the difference between $\tilde\lambda$ and $\arcsin\tilde\lambda$.
With $d$ applications of $V$, we obtain an estimate of $\arcsin\tilde\lambda$ with standard deviation $O(1/d)$ (see for example \cite[Theorem 5]{Childs08}).
Since $\tilde\lambda-\arcsin\tilde\lambda = O(\tilde\lambda^3)$, this estimate only improves the accuracy if $1/d$ is small compared to $\tilde\lambda$.  If $\tilde\lambda d < 1$, we do not perform a correction.  In that case the error is $O(d \tilde\lambda^3) \le O(\tilde\lambda^2)$.
On the other hand, if $\tilde\lambda d \ge 1$, then the standard deviation in the estimate of $\tilde\lambda-\arcsin\tilde\lambda$ is $O(\tilde\lambda^2/d)$, and again the error in the final phase is $O(\tilde\lambda^2)$.
So in both cases, the error in the final state is $O(\tilde\lambda^2) \le O(\La^2\unlaziness^2/\Li^2)$. The overall number of steps of the quantum walk used is $O(d) = O(\Li t/\unlaziness)$.
\end{proof}

Combining the results of Lemmas \ref{lemma1} and \ref{lemma2} gives the improved Hamiltonian simulation method described by Theorem \ref{theorem1}.

\begin{proof}[Proof of Theorem \ref{theorem1}.]
We apply the Hamiltonian simulation described in Lemma \ref{lemma2}, with the method described in Lemma \ref{lemma1} to perform the steps of the quantum walk.

Using
\begin{equation}
\label{eq:xval}
X= \frac 1{\sparseness t}\max \left\{ \lceil {\La t}/{\sqrt\delta} \rceil,\left\lceil \Lm\sparseness t \right\rceil \right\},
\end{equation}
we take $\unlaziness=\Li/\sparseness X$. This value of $X$ satisfies $X\ge\Lm$ (as we always require for $X$), which implies that the condition $\unlaziness \le \Li/\sparseness \Lm$ of Lemma \ref{lemma1} is satisfied. Thus, by Lemma \ref{lemma1}, the steps of the quantum walk can be performed using $O(1)$ queries.

With this value of $\unlaziness$, $\Li t/\unlaziness=X\sparseness t$ is an integer, and we can use Lemma \ref{lemma2}. Then the error is
\begin{align}
O(\La^2\unlaziness^2/\Li^2) &= O(\La^2/\sparseness^2 X^2) \le O(\error).
\end{align}
The total number of oracle calls is $O(\Li t/\unlaziness)=O(X\sparseness t)$, and is therefore
\begin{equation}
O(\lceil {\La t}/{\sqrt\delta} \rceil + \left\lceil \Lm\sparseness t \right\rceil) =
O({\La t}/{\sqrt\delta} + \Lm\sparseness t + 1)
\end{equation}
as claimed.
\end{proof}
Comparing this to the simulation of sparse Hamiltonians using high-order integrators \cite{Berry07}, the scaling is better in terms of all parameters except $\delta$. (We assume that the upper bounds on norms of $H$ have the same order as the norms themselves, so for example $\La=O(\Ha)$.) The number of queries is only linear in $\Ha t$, as opposed to slightly superlinear. The scaling is particularly improved in terms of $\sparseness$, as it is only linear, whereas the scaling in Ref.\ \cite{Berry07} was as $\sparseness^4$. The scaling in $\error$ is as $1/\sqrt\error$, as opposed to an arbitrarily small power in Ref.\ \cite{Berry07}. However, there is an advantage in that for $\delta = O(1/\sparseness^2)$ there is no further explicit dependence on $D$.

This also improves over the method of Ref.\ \cite{Childs08}, which uses $O(\Hi t/\delta)$ steps of the quantum walk. Using a naive method for state preparation---simply querying all nonzero elements---uses $O(\sparseness\Hi t/\delta)$ queries. Theorem \ref{theorem1} improves on this as $\Hi$ is typically larger than both $\Hm$ and $\Ha$, and because the scaling with $\error$ is improved. In particular, the bound $\Hi\le \sparseness\Hm$ shows that $\sparseness\Hi t/\delta \le O(\sparseness^2\Hm t/\error)$, which is worse than $\sparseness \Hm t$, and the bound $\Hi\le \sqrt{\sparseness}\Ha$ shows that $\sparseness\Hi t/\delta \le O(\sparseness^{3/2}\Ha t/\error)$, which is worse than $\Ha t/\sqrt\error$.

We conclude this section by describing in more detail how the isometry $T$ is used in each part of the simulation. It is used in three different ways:
\begin{enumerate}
\item At the beginning, to map the initial state from $\CC^\HN$ into the tensor product space $\CC^{2\HN} \otimes \CC^{2\HN}$.
\item To implement each application of $V$.
\item At the end, to map the final state from $\CC^{2\HN} \otimes \CC^{2\HN}$ back to $\CC^\HN$.
\end{enumerate}

For the first step, we can directly implement $T$ as defined in Eq.\ \eqref{eq:tdef}. The initial state is a superposition of states $\ket{j}$ for $j \in \{1,2,\ldots,\HN\}$ used to control the state preparation. We simply introduce another register in which $\ket{\phib_j}$ is prepared.

When $T$ is used to implement $V$, we need to specify the action on the extra qubit introduced in the state preparation procedure.  We only require the correct eigenvalue and eigenvector relations, namely, that
\begin{equation}
\label{eq:necrel}
T^\dagger ST \ket{\lambda} = \tilde\lambda\ket{\lambda}.
\end{equation}
On the expanded space, the isometry $T$ should have the form
\begin{equation}
\sum_{j=1}^{\HN} \left[ \ket{j,0}\ket{\phib_j}\bra{j,0}+\ket{j,1}\ket{\Omega_j}\bra{j,1} \right]
\end{equation}
for some states $\ket{\Omega_j}$.
Because $\ket{\lambda}$ is orthogonal to $\ket{j,1}$ (as the initial state has the qubit initiali{\s}ed as $\ket{0}$),
\begin{equation}
T\ket{\lambda} = \sum_{j=1}^{\HN} \ket{j,0}\ket{\phib_j}\braket{j,0}{\lambda},
\end{equation}
giving
\begin{align}
T^\dagger ST \ket{\lambda} &= \sum_{j,k=1}^{\HN} \big[\braket{j,0}{\phib_k}\braket{\phib_j}{k,0}\ket{j,0}\braket{k,0}{\lambda} 
\nonumber \\
&\quad +\braket{j,1}{\phib_k}\braket{\Omega_j}{k,0}\ket{j,1}\braket{k,0}{\lambda}\big].
\end{align}
To obtain Eq.\ \eqref{eq:necrel}, we simply need the second term above to vanish. This can be ensured by taking $\ket{\Omega_j}=\ket{\Omega,1}$ for any state $\ket{\Omega}$. 
Note that it is important to properly apply the operator $V$ to states where the ancilla qubit is in the state $\ket{1}$, since although the individual $T\ket{\lambda}$ have the ancilla in the $\ket{0}$ state, the eigenstates $\ket{\mu_\pm^\lambda}$ of $V$ have a component of $ST\ket{\lambda}$, and therefore have a component with the ancilla in the $\ket{1}$ state.

For the final use of $T$, we wish to map the state $\ket{j,0}\ket{\phib_j}$ to $\ket{j,0}$ for each $j \in \{1,2,\ldots,\HN\}$. In general, this can only be carried out approximately, because the final state will not be exactly a superposition of states of the form $\ket{j,0}\ket{\phib_j}$. First, if the ancilla qubit is in the state $\ket{1}$, this may be regarded as a failure, because the ideal final state has the ancilla in the state $\ket{0}$. Otherwise, we perform inverse state preparation conditional on the index $j$. Starting from $\ket{j,0}\ket{\phib_j}$, the second register should ideally be mapped to the initial state used for the state preparation; any other state can be regarded as failure. 
In general, the total failure probability is proportional to the error in the inverse state preparation procedure.

\section{Improved state preparation}
\label{sec:impl}

To further improve our simulations, especially in the non-sparse case, we consider state preparation techniques based on amplitude amplification \cite{Brassard97,Brassard,Grover98}.

Our techniques draw from work on black-box state preparation.  In that problem, we are given an oracle $O_\psi$ acting as
\begin{equation}
O_\psi \ket{j}\ket{z} = \ket{j}\ket{z \oplus \psib_j}
\end{equation}
for some quantum state $\ket\psib = \sum_{j=1}^\HN \psib_j \ket{j}$; the goal is to prepare a copy of $\ket\psib$.
Grover showed how to prepare a black-box quantum state with only $O(\sqrt{\HN})$ queries \cite{Grover00}. By the lower bound for search \cite{Grovopt}, preparation of an $\HN$-dimensional black-box quantum state requires $\Omega(\sqrt{\HN})$ queries, so this state preparation scheme is optimal.

We can use Grover's technique to prepare the states from Eq.\ \eqref{eq:alt3} and thereby implement the quantum walk.  This is favorable for large $\sparseness$, in which case state preparation based on Lemma \ref{lemma1} alone is suboptimal.
By combining the approach of Lemma \ref{lemma1} with amplitude amplification, 
we improve on both these approaches, as follows.
\begin{lemma}
\label{lemma4}
Let $\Li\ge\Hi$, $\Lm\ge\Hm$, and $\unlaziness\in(0,1]$.
Then
\begin{equation}
\label{eq:alt4}
\ket{\phib_j} = \sqrt{\frac{\unlaziness}{\Li}}\sum_{k=1}^{\HN}\sqrt{H_{jk}^*}\ket{k}\ket{0}+\sqrt{1-\frac{\unlaziness\Sj}{\Li}}\ket{\psic_j}\ket{1}
\end{equation}
can be approximately prepared using
\begin{equation}
\label{eq:calls}
O\left(\sqrt{\frac{\unlaziness \Lm \sparseness}{\Li}}+1\right)
\end{equation}
queries to $O_H$ and $O_\F$.
The approximation has relative error in the weighting of the first term of $O(\unlaziness)$.
\end{lemma}

\begin{proof}
As in the proof of Lemma \ref{lemma1}, we can prepare
\begin{align}
\label{eq:first}
\ket{\phib^b_j} &\defeq \frac{1}{\sqrt{\sparseness}}\sum_{k\in \F_j}
\ket{k}\left[\sqrt{\frac{H_{jk}^*}{X}}\ket{0}  +\sqrt{1-\frac{|H_{jk}|}{X}}\ket{1}\right]
\end{align}
using one query to $O_\F$ and two queries to $O_H$, where $X$ is a real number satisfying $X\ge\Lm\ge\Hm$. Let $B_j$ denote a unitary operation that prepares $\ket{\phib^b_j}$ from $\ket{0}\ket{0}$.

We now use a form of amplitude amplification similar to that introduced by Grover \cite{Grover00}. We define two reflection operators. The first reflects about the $\ket{0}$ state for the ancilla qubit,
\begin{equation}
R^f \defeq \openone \otimes (\openone - 2\ket{0}\bra{0}),
\end{equation}
and the second reflects about the state $\ket{\phib^b_j}$,
\begin{equation}
R^b_j \defeq 2\ket{\phib_j^b}\bra{\phib_j^b} - \openone.
\end{equation}
The latter reflection can be performed by applying $B_j^\dagger$, reflecting about $\ket{0}\ket{0}$, and then applying $B_j$. Using an appropriate number of these reflections, we could obtain a final state close to
\begin{equation}
\label{eq:phijf}
\ket{\phib^f_j} \defeq \frac{1}{\sqrt{\Sj}}\sum_{k=1}^{\HN}\sqrt{H_{jk}^*}\ket{k}\ket{0}.
\end{equation}
However, the key point is that we do not rotate all the way towards this state, but instead prepare 
\begin{equation}
\label{eq:phij}
\ket{\phib_j} = \sqrt{\frac{\unlaziness \Sj}{\Li}} \ket{\phib^f_j} + {\cal N}_j \sum_{k\in \F_j}\sqrt{1-\frac{|H_{jk}|}{X_j}}\ket{k}\ket{1},
\end{equation}
where ${\cal N}_j$ is a normali{\s}ation constant. This expression corresponds to the definition of $\ket{\phib_j}$ with
\begin{equation}
\ket{\psic_j} \propto \sum_{k\in \F_j}\sqrt{1-\frac{|H_{jk}|}{X_j}}\ket{k}.
\end{equation}
According to Eq.\ \eqref{eq:normalizedh}, the normali{\s}ed Hamiltonian corresponding to the discrete-time quantum walk defined by these states is
\begin{equation}
\label{eq:h0def}
\tilde H=\frac{\unlaziness H}{\Li}.
\end{equation}

We prepare a state close to $\ket{\phib_j}$ using amplitude amplification.  Let
\begin{equation}
\ket{\phib_j(r)} \defeq (R_j^b R^f)^r \ket{\phib_j^b}
\end{equation}
denote the state as a function of the number of steps, $r$.  We have
\begin{align}
\label{eq:phir}
&\ket{\phib_j(r)} = \sin[(2r+1)\theta_j] \ket{\phib^f_j} \nonumber \\
& \qquad + {\cal N}_j \cos[(2r+1)\theta_j]\sum_{k\in \F_j}\sqrt{1-\frac{|H_{jk}|}{X}}\ket{k}\ket{1},
\end{align}
where
\begin{equation}
\sin\theta_j = \braket{\phib^f_j}{\phib^b_j} = \sqrt{\frac{\Sj}{\sparseness X}}.
\end{equation}
By Eqs.\ \eqref{eq:phijf} and \eqref{eq:phij}, $\braket{\phib^f_j}{\phib_j}=\sqrt{\unlaziness \Sj/\Li}$, so the value of $r$ that gives the desired outcome is
\begin{align}
\label{eq:rjopt}
r_j^{\rm opt} &\defeq \frac{1}{2} \left( \frac{1}{\theta_j} \arcsin\sqrt{\frac{\unlaziness \Sj}{\Li}}-1\right) \nonumber\\
& =\frac{1}{2} \left( \frac{\arcsin\sqrt{\frac{\unlaziness \Sj}{\Li}}}{\arcsin\sqrt{\frac{\Sj}{\sparseness X}}} -1\right).
\end{align}

There are several reasons why we cannot perform exactly $r_j^{\rm opt}$ queries.  This value may not be an integer, and it is $j$-dependent.  Furthermore, since $\Sj$ is not known in general, the exact value of $r_j^{\rm opt}$ is unknown.
However, if $\unlaziness$ is small, then the $\arcsin$ function can be lineari{\s}ed, and we can take
\begin{equation}
r \approx \frac{1}{2} \sqrt{\frac{\unlaziness X \sparseness}{\Li}}-\frac{1}{2}.
\end{equation}
Specifically, we choose 
\begin{equation}
r = \left\lceil \frac{1}{2} \sqrt{\frac{\unlaziness \Lm \sparseness}{\Li}}-\frac{1}{2} \right\rceil
\end{equation}
and
\begin{equation}
\label{eq:xval2}
X=(2r+1)^2 \frac{\Li}{\unlaziness\sparseness}.
\end{equation}
Since the number of queries per step is $O(1)$, the total number of queries is $O(\sqrt{\unlaziness\Lm\sparseness/\Li}+1)$ as claimed.

Now we analy{\s}e the error incurred due to imperfect state preparation.  First consider the deviation of $r$ from $r_j^{\rm opt}$.  This deviation results from lineari{\s}ation of both the numerator and denominator of Eq.\ \eqref{eq:rjopt}. The argument of the $\arcsin$ function in the denominator is smaller than that in the numerator,
so to determine the scaling of the error, it suffices to consider the error in the lineari{\s}ation of the numerator. The relative error is thus
\begin{equation}
\frac{|r_j^{\rm opt} - r|}r = O(\unlaziness\sigma_j/\Li) \le O(\unlaziness)
\end{equation}
since $\Sj \le \Li$ for all $j$.

The effect of the difference between $r$ and $r_j^{\rm opt}$ is a slightly incorrect weighting of $\ket{\phib_j^f}$ in the final state:
\begin{align}
\braket{\phib_j^f}{\phib_j(r)} 
&= \sin[(2r+1)\theta_j] \nn
&= \sqrt{\frac{\unlaziness\sigma_j}{\Li}}(1+x_j)
\end{align}
where
\begin{align}
x_j
&\defeq \sqrt{\frac{\Li}{\unlaziness\sigma_j}} \sin[(2r+1)\theta_j] - 1 \nn
&= \sqrt{\frac{\Li}{\unlaziness\sigma_j}} \left\{\sin[(2r_j^{\rm opt}+1)\theta_j]\right. \nn & \quad \left. +2\theta_j\cos[(2r_j^{\rm int}+1)\theta_j](r-r_j^{\rm opt})\right\}-1 \nn
&= \sqrt{\frac{\Li}{\unlaziness\sigma_j}}2\theta_j\cos[(2r_j^{\rm int}+1)\theta_j](r-r_j^{\rm opt})
\end{align}
for some $r_j^{\rm int} \in [r,r_j^{\rm opt}]$, where in the second line we have used Taylor's theorem. Hence
\begin{align}
|x_j| &\le \sqrt{\frac{\Li}{\unlaziness\sigma_j}}2\theta_j|r-r_j^{\rm opt}| \nn
&\le \sqrt{\frac{\Li}{\unlaziness\sigma_j}}\pi\sqrt{\frac{\sigma_j}{DX}}|r-r_j^{\rm opt}| \nn
&= \frac {\pi r}{2r+1}\frac{|r-r_j^{\rm opt}|}r \nn
&= O(\unlaziness).
\end{align}
In the next to last line, we have used Eq.\ \eqref{eq:xval2}. Hence the error in the weighting of the first term in Eq.\ \eqref{eq:alt4} is $O(\unlaziness)$, as claimed.
\end{proof}

The state preparation scheme described in Lemma \ref{lemma4} introduces additional error in the Hamiltonian simulation, but this error is well bounded. In particular, we have the following.
\begin{lemma}
\label{lemma5}
The error in the state preparation scheme of Lemma \ref{lemma4} results in an error in the Hamiltonian simulation described in Sec.\ \ref{sec:hamsim} of $O(\Ha t\unlaziness)$.
\end{lemma}

\begin{proof}
The actual Hamiltonian being simulated, $\tilde H'$, has matrix elements
\begin{align}
\tilde H'_{jk}
&= \bra{j,0}\bra{\phib_j(r)}S\ket{k,0}\ket{\phib_k(r)} \nn
&= \braket{\phib_j(r)}{k,0}\braket{j,0}{\phib_k(r)} \nn
& = \sin[(2r+1)\theta_j]\sin[(2r+1)\theta_k]\frac{H_{jk}}{\sqrt{\sigma_j\sigma_k}} \nn
& = \frac{\unlaziness H_{jk}}{\Li}(1+x_j)(1+x_k).
\end{align}
Defining a diagonal matrix ${\bf{x}} \defeq \diag(x_1,x_2,\ldots,x_\HN)$, the error in the Hamiltonian is
\begin{align}
\| \tilde H' - \tilde H \|
&= \left\|\frac{\unlaziness}{\Li}({\bf{x}}H+H{\bf{x}}+{\bf{x}}H{\bf{x}})\right\| \nn
&\le \frac{\unlaziness \Ha}{\Li}(2 x_{\rm max}
+x_{\rm max}^2) \nn
&= \frac{\unlaziness \Ha}{\Li}O(\unlaziness),
\end{align}
where $x_{\rm max}\defeq\max_j |{x_j}|$. In evolving the Hamiltonian over time $t$, we multiply this by a factor of $t\Li/\unlaziness$, so the resulting error is $O(\Ha t\unlaziness)$.
\end{proof}

\section{Non-sparse Hamiltonians}
\label{sec:full}

Now we examine the overall performance of the Hamiltonian simulation algorithm with improved state preparation. Multiplying the number of steps of the quantum walk by the number of queries required to implement each step, we find the following.
\begin{lemma}
\label{lemma6}
Given a black-box Hamiltonian H, let $\La\ge\Ha$, $\Li\ge\Hi$, and $\Lm\ge\Hm$.  Then $H$ can be simulated for time $t$ with error at most $\error\in(0,1]$ using
\begin{equation}
\label{eq:lem6}
O\left( t^{3/2}\sqrt{\frac{\Lm\sparseness\Li\La}{\error}}\right)
\end{equation}
queries to $O_H$ and $O_\F$, provided that
\begin{align}
\label{eq:rests}
\La t &\ge \sqrt{\delta}, \\
\label{eq:rests2}
\La t &\ge \frac{\La^2}{\Lm\Li\sparseness}, ~ \text{and} \\
\label{eq:helpful}
\La &\le \Li.
\end{align}
\end{lemma}

The restriction \eqref{eq:helpful} simply means that $\La$ is not unnecessarily large. Because $\Ha\le\Hi$, we can decrease any given $\La$ to be at most $\Li$, provided \eqref{eq:rests} still holds.
This Lemma provides improved performance in cases where $\sparseness$ is large.
This may mean that $\sparseness=\HN$, but we continue to perform the analysis in terms of the sparseness parameter $\sparseness$ for generality.

\begin{proof}
We take
\begin{equation}
\label{eq:unval}
\unlaziness=\frac{\Li t}{\lceil \Li\La t^2/\error\rceil}.
\end{equation}
This ensures that $\unlaziness\le \error/\La t$, so $\Ha t\unlaziness\le\error$. The restriction \eqref{eq:rests} then ensures that $\unlaziness\le 1$. In addition, \eqref{eq:rests} and \eqref{eq:helpful} ensure that $\Li\La t^2/\error\ge 1$, so the ceiling function does not affect the scaling, and
\begin{equation}
\label{eq:unscaling}
1/\unlaziness = O(\La t/\error).
\end{equation}

With this value of $\unlaziness$, $\Li t/\unlaziness$ is an integer, and therefore Lemma \ref{lemma2} shows that $O(\Li t/\unlaziness)$ quantum walk steps suffice for the simulation. Then, using Lemma \ref{lemma4}, the number of oracle queries for each step of the quantum walk is $O(\sqrt{\unlaziness\Lm\sparseness/\Li})$, unless this quantity is less than 1, in which case the state preparation proceeds without amplitude amplification.

That case does not alter the result, because the total number of queries for the simulation as given by Eq.\ \eqref{eq:other} in Theorem \ref{theorem1} is less than Eq.\ \eqref{eq:lem6} given the restrictions in Lemma \ref{lemma6}.  This can be shown as follows. First, assuming $\sqrt{\unlaziness\Lm\sparseness/\Li}=O(1)$, we have
\begin{align}
\sparseness\Lm t &= \frac{\sqrt{\error}\sparseness\Lm t}{\sqrt{\error}} \nn
&= O\left( \frac{\sqrt{\unlaziness\La t}\sparseness\Lm t}{\sqrt{\error}} \right) \nn
&\le O\left( t^{3/2}\sqrt{\frac{\Lm\sparseness\Li\La}{\error}}\right).
\end{align}
In the second line we have used Eq.\ \eqref{eq:unscaling}, and in the third line we have used the condition that $\sqrt{\unlaziness\Lm\sparseness/\Li}=O(1)$.
Next,
\begin{equation}
\frac{\La t}{\sqrt\error} 
\le t^{3/2}\sqrt{\frac{\Lm\sparseness\Li\La}{\error}}
\end{equation}
using the restriction \eqref{eq:rests2}.
Finally, combining Eqs.\ \eqref{eq:rests} and \eqref{eq:rests2} shows that the number of queries in Eq.\ \eqref{eq:lem6} is at least constant.  Thus we find that Eq.\ \eqref{eq:other} is less than Eq.\ \eqref{eq:lem6}, as required.

For the case where state preparation proceeds via amplitude amplification, we multiply the number of steps of the quantum walk (from Lemma \ref{lemma2}) by the number of oracle calls for each step (from Lemma \ref{lemma4}). Thus the total number of queries is
\begin{equation}
\label{eq:nocalls}
O\left(t\sqrt{\frac{\Lm\sparseness\Li}{\unlaziness}}\right)
\le
O\left(t^{3/2}\sqrt{\frac{\Lm\sparseness\Li\La}{\error}}\right).
\end{equation}
where we have used Eq.\ \eqref{eq:unscaling}.

Finally, we consider the error in the simulation. Because $\unlaziness\le\error/\La t$, Lemma \ref{lemma5} implies that the error due to imperfect state preparation is $O(\error)$. Using Lemma \ref{lemma2}, the error due to the quantum walk simulation is $O(\La^2\unlaziness^2/\Li^2)$. Using $\unlaziness\le \error/\La t$ and $\sqrt{\error} \le \La t$, this contribution to the error is also $O(\error)$.

The statement of the Lemma requires that the error is less than $\delta$, rather than $O(\delta)$. However, any multiplying factor for the error can be absorbed into the big-$O$ notation of Eq.\ \eqref{eq:lem6}.
\end{proof}

We are interested in improving the scaling with $\sparseness$ beyond the linear scaling in Theorem \ref{theorem1}.
The number of queries in Lemma \ref{lemma6} contains $\sqrt{\sparseness}$, but also depends on several other quantities.
For simplicity, in this discussion we assume that $\La$ can be replaced with $\Ha$, and so forth.
In the worst case we can have $\Hm\propto\Ha$ and $\Hi\propto\Ha\sqrt{\sparseness}$.
This would yield overall scaling of $O((\Ha t)^{3/2}\sparseness^{3/4}/\sqrt\error)$.
However, it should be noted that this worst case arises from two different factors.
\begin{enumerate}
\item To have $\Hm\propto\Ha$, the distribution of the magnitudes of the matrix elements should have a sharp peak, so there is a row with most of the weight on one of the elements.
\item To have $\Hi\propto\Ha\sqrt{\sparseness}$, the magnitudes of the matrix elements should be relatively evenly distributed.
\end{enumerate}
If we could ensure that all the nonzero elements had magnitudes within some constant factor (so there is no sharp peak), then we would obtain $\Hm\propto\Ha/\sqrt{\sparseness}$, giving a scaling of $O((\Ha t)^{3/2}\sqrt{\sparseness/\error})$.

\section{Breaking up the Hamiltonian}
\label{sec:break}

We now consider how the simulation can be improved by breaking up the Hamiltonian into a sum of terms.
Although the matrix elements of the Hamiltonian may differ over a wide range, the Hamiltonian can be broken up into terms, each of which has matrix elements of similar magnitude.
By combining the evolution under these Hamiltonians via a Lie-Trotter-Suzuki formula, we can expect scaling close to $O((\Ha t)^{3/2}\sqrt{\sparseness/\error})$.
The only problem is that the spectral norms of the individual Hamiltonians may be large.
First we present a derivation showing that, provided the norms of the individual terms are not large, then the expected scaling is obtained.
Next we present numerical results showing that typical spectral norms are small, although there are pathological cases with large norms. Finally, we present a general method using a number of queries roughly proportional to $\sparseness^{2/3}$ even when the spectral norms are large.

\subsection{Small norms}
\label{sec:smallnorms}

In order to present our result, we define the function ``\brk'', which quantifies how much the norm can be increased by breaking up the Hamiltonian into parts.  Let
\begin{equation}
\brk(H) := \max_{a,b\in \mathbb{R}} \| H^{ab} \| / \|H\|,
\end{equation}
where the matrix $H^{ab}$ is defined by
\begin{equation}
H^{ab}_{jk} := \begin{cases} H_{jk} & \text{if $a < |H_{jk} | \le b$}, \\
0 & \text{if $|H_{jk} | \le a {\rm ~or~} b<|H_{jk} |$}. \end{cases}
\end{equation}
In this subsection we suppose that $\brk(H)$ is small.
We present numerical evidence in Sec.\ \ref{sec:norms} that $\brk(H)\le 1.5$ in most cases.
From the definition, it is clear that $\brk(H)\ge 1$.
In addition, because $\|H^{ab}\|\le \|H^{ab}\|_1\le\|H\|_1\le\|H\|\sqrt{\sparseness}$, we have $\brk(H)\le\sqrt{\sparseness}$.
If $\brk(H)$ can be upper bounded by a constant, we obtain a simulation with scaling close to $\sqrt\sparseness$.

\begin{theorem}
\label{theoremsm}
Let $\La\ge\Ha$ and $\sca\in[\brk(H),\sqrt{\sparseness}]$.  The evolution under the Hamiltonian $H$ for time $t$ can be simulated with error at most $\error\in(0,1]$ using
\begin{equation}
\label{ex:smres}
O\left( \sqrt{\sca\sparseness/\error} (\log\sparseness)^{7/4} (\La t)^{3/2} \right)
\end{equation}
queries to $O_H$ and $O_\F$, provided $\error\sparseness > \La t > \sqrt\error$.
\end{theorem}

\begin{proof}
We split the Hamiltonian into $L$ terms, each with nonzero elements of approximately the same magnitude:
\begin{equation}
\label{eq:split}
H = \sum_{\hj=1}^{\hm} H_\hj.
\end{equation}
We take the Hamiltonians $H_\hj$ to include elements with decreasing magnitudes: $H_1$ contains elements with the largest magnitudes, $H_2$ contains elements with the next largest magnitudes, and so forth. We denote the cutoff values $A_\hj$, so $H_\hj = H^{A_\hj A_{\hj-1}}$ for $\hj < \hm$ and $H_\hm = H^{0A_{\hm-1}}$.
We take $A_0=\La$, $A_\hm=\La/\sqrt{\sparseness}$, and $A_0>A_1>\cdots>A_\hm$.
In examining $H_\hj$, let $\La^{(\hj)}$ denote an upper bound on $\|H_\hj\|$, $\Li^{(\hj)}$ an upper bound on $\|H_\hj\|_1$, $\Lm^{(\hj)}$ an upper bound on $\|H_\hj\|_{\rm max}$, and $\tau_\hj$ the time interval for simulation of $H_\hj$. We also let $\error_\hj$ denote the error allowed for simulating $H_\hj$ over a time step of length $\tau_\hj$.

Because $|[H_\hj]_{jk}|\le A_{\hj-1}$, we can take $\Lm^{(\hj)} = A_{\hj-1}$. To choose a value of $\Li^{(\hj)}$ for $\hj<\hm$, we use
\begin{align}
\|H_\hj\|_1 &\le \max_{j} \sum_{k=1}^M |[H_\hj]_{jk}|^2/A_\hj \nn
&\le \max_{j} \sum_{k=1}^M |H_{jk}|^2/A_\hj \nn
&\le \Ha^2/A_\hj .
\end{align}
Therefore we can take $\Li^{(\hj)}=\La^2/A_\hj$ for $\hj<\hm$. For $\hj=\hm$, we have 
\begin{equation}
\|H_\hj\|_1 \le \Hi \le \Ha\sqrt{\sparseness}.
\end{equation}
Since we set $A_\hm=\La/\sqrt{\sparseness}$, we have $\Li^{(\hj)}=\La^2/A_\hj$ for $\hj=\hm$ as well. This is why we define a value for $A_\hm$, even though it is not used to bound matrix elements.

The success of the simulation depends crucially on the scaling of the norms $\|H_\hj\|$. 
By assumption, $\sca\ge\brk(H)$, so $\|H_\hj\|\le \sca\Ha$.
Because $\|H_\hj\| \le \|H_\hj\|_1$, we can take $\La^{(\hj)}=\min\{\sca\La,\La^2/A_\hj\}$.

Using Lemma \ref{lemma6}, the number of queries to simulate $H_\hj$ for time $\tau_\hj$ is
\begin{equation}
\label{eq:witham}
O\left(\tau_\hj^{3/2} \sqrt{\frac{\Lm^{(\hj)}\sparseness\Li^{(\hj)}\La^{(\hj)}}{\error_\hj}}\right) \le
O\left(\La\tau_\hj^{3/2}\sqrt{\frac{\sparseness \sca\La A_{\hj-1}}{\error_\hj A_\hj}}\right).
\end{equation}
In this proof we take $\tau_\hj$ and $\error_\hj$ to be independent of $\hj$. To ensure that the number of queries is independent of $\hj$ (so no one term dominates the scaling), we take constant ratios $A_{\hj-1}/A_\hj$. To satisfy $A_0=\La$ and $A_\hm=\La/\sqrt{\sparseness}$, we can take $A_\hj=\La\sparseness^{-\hj/2\hm}$.
Then $\Lm^{(\hj)}=A_{\hj-1}=\La\sparseness^{(1-\hj)/2\hm}$, $\Li^{(\hj)}=\La^2/A_{\hj}=\La\sparseness^{\hj/2\hm}$, and the ratio between successive cutoffs is $A_{\hj-1}/A_\hj = \sparseness^{1/2\hm}$.

Next we ensure that the conditions of Lemma \ref{lemma6} hold. Condition \eqref{eq:helpful} follows immediately from $\La^{(\hj)}=\min\{\sca\La,\La^2/A_\hj\}$. To satisfy conditions \eqref{eq:rests} and \eqref{eq:rests2}, $\tau_\hj$ cannot be too small, but it must be small enough that the Trotter error is $O(\error)$. To achieve this, we choose $\tau_\hj$ to satisfy
\begin{equation}
\label{eq:tl}
\tau_\hj \ge \max \left\{ \frac{\error}{\hm^{3/2}\La^2 t} , \frac{\sca}{\La\sparseness^{1+1/2\hm}} \right\}.
\end{equation}
In addition, to apply the Trotter formula, $t/\tau_\hj$ must be an \emph{even} integer. Thus we take
\begin{equation}
\label{eq:tlb}
\tau_\hj = \frac{t}{2 \left\lfloor\min \left\{ \frac{\hm^{3/2}\La^2 t^2}{2\error} , \frac{\La\sparseness^{1+1/2\hm}t} {2\sca}\right\}\right\rfloor}.
\end{equation}

For this expression to be well-defined, the denominator must be nonzero. For the first term of the minimum, we find that
\begin{equation}
\frac{\hm^{3/2}\La^2 t^2}{2\error} > \frac{\hm^{3/2}}{2} > 1.
\end{equation}
The first inequality uses the condition $\La t > \sqrt\error$ and the second uses $\hm \ge 2$ (since otherwise we are not breaking up the Hamiltonian at all). For the second term to be at least 1, we require
\begin{equation}
\label{eq:vio}
\La\sparseness t \ge 2\sca \sparseness^{-1/2\hm}.
\end{equation}
If this does not hold, then we perform the simulation with Theorem \ref{theorem1} instead of Lemma \ref{lemma6}. Since $\La\ge\Hm$, we can take $\Lm=\La$. The condition $\error\sparseness > \La t > \sqrt\error$ implies that Eq.\ \eqref{eq:other} is $O(\sparseness\La t)$. Provided Eq.\ \eqref{eq:vio} is violated, we find that we can simulate the Hamiltonian with $O(\sparseness^{1/2\hm})$ queries. We will take $\hm \propto \log \sparseness$, so the simulation uses $O(1)$ queries, which is no more than Eq.\ \eqref{ex:smres}.
Thus, for the remainder of this proof, we assume that Eq.\ \eqref{eq:vio} holds, so Eq.\ \eqref{eq:tlb} is well-defined.

Using Eq.\ \eqref{eq:tlb}, we find that Eq.\ \eqref{eq:tl} is satisfied, and
\begin{equation}
\La^{(\hj)} \tau_\hj \ge \La \tau_\hj \ge \La \sqrt{\tau_\hj}\sqrt{\frac{\error}{\hm^{3/2}\La^2 t}} = \sqrt{\frac{\error\tau_\hj}{\hm^{3/2} t}}.
\end{equation}
The first inequality uses $\La^{(\hj)}\ge \La$ and the second uses Eq.\ \eqref{eq:tl}.
Taking
\begin{equation}
\error_\hj=\frac{\error\tau_\hj}{\hm^{3/2} t},
\end{equation}
we obtain $\La^{(\hj)} \tau_\hj \ge \sqrt{\error_\hj}$, so Eq.\ \eqref{eq:rests} is satisfied.
Equation \eqref{eq:rests2} follows from
\begin{equation}
\La^{(\hj)} \tau_\hj \ge \La^{(\hj)}\sca/\La\sparseness^{1+1/2\hm} \ge (\La^{(\hj)})^2/\Lm^{(\hj)}\Li^{(\hj)}\sparseness.
\end{equation}
Here the first inequality holds due to the second term of the maximum in Eq.\ \eqref{eq:tl}.

Now we use a $K$th order Lie-Trotter-Suzuki integrator to combine the simulations of the $H_\hj$ into a simulation of $H$.
The Strang splitting formula \cite{Strang} corresponds to $K=1$; larger values of $K$ correspond to higher-order Lie-Trotter-Suzuki formulae. In this proof we simply take $K=1$; in Section \ref{sec:largenorms} we will consider the case $K=2$.
The simulation resulting from a $K$th order integrator is approximate, introducing error \cite{Suzuki90,Suzuki91,Berry07}
\begin{equation}
\label{eq:inter}
O\left( \left[2\hm 5^{K-1}\tau\max_{\hj}\La^{(\hj)}\right]^{2K+1}\frac{t}{\tau} \right).
\end{equation}
Here $\tau$ is the time interval over which the integrator is repeated. That is, the time is broken up into $t/\tau$ intervals, and the same integrator is used on each of those intervals. By Eq.\ (A2) of Ref.\ \cite{Nathan10}, $\tau_\hj \ge \tau \times (3/2) 3^{-K}$. Thus, for a fixed value of $K$, $\tau = O(\tau_\hj)$. Also, the number of queries is increased by a factor of $5^K \hm$ due to the number of terms in the integrator.

To bound the error, we need to take account of the error due to the individual simulations and the error due to the Trotter formula.
There are $O(\hm t/\tau)$ terms in the Trotter formula for $K=1$, so the total error in performing the individual simulations (neglecting only the error introduced by the Trotter formula) is $O(\error_\hj\hm t/\tau)$. Because we take $\error_\hj=\error\tau_\hj/\hm^{3/2} t$, the total error due to the simulations is $O(\error/\hm^{1/2})$, which is $O(\error)$. 

With $K=1$ and $\hm\propto\log\sparseness$, the Trotter error from Eq.~\eqref{eq:inter} is $O(L^3 \La^3 \tau_\hj^2 t)$.
If $\error/\hm^{3/2}\La t \ge \sca/\sparseness^{1-1/2\hm}$, so that Eq.~\eqref{eq:tlb} gives $\tau_\hj=O(\error/\hm^{3/2}\La^2 t)$, then this Trotter error is $O(\error^2/\La t)$,
which is $O(\error)$ because $\La t> \error$.
Alternatively, if $\error/\hm^{3/2}\La t < \sca/\sparseness^{1+1/2\hm}$, then the Trotter error is
$O(\La t \hm^3\sca^2/\sparseness^{2+1/\hm})$,
which is $O(\error \hm^3\sca^2/\sparseness^{1+1/\hm})$. With $\hm\propto \log\sparseness$ and $\sca\le\sqrt{\sparseness}$, the Trotter error is $O(\error)$.

The total number of queries is given by \eqref{eq:witham} multiplied by $\hm t/\tau$, so we obtain a simulation using
\begin{equation}
\label{eq:amsca}
O\left( \hm^{7/4} (\La t)^{3/2} \sparseness^{1/2+1/(4\hm)}\sqrt{\sca/\error} \right)
\end{equation}
queries.
Now taking $\hm\propto \log\sparseness$, $\sparseness^{1/\hm}=O(1)$, so the number of queries is as given in Eq.\ \eqref{ex:smres}.
The condition $\error \sparseness > \La t > \sqrt\error$ ensures that this is at least 1.
\end{proof}

This theorem holds regardless of whether the norms are small.
In the worst case we can have $\sca=\sqrt\sparseness$, in which case the $\sparseness^{3/4}$ scaling is again obtained (as at the end of Sec.\ \ref{sec:full}).
On the other hand, if breaking up the Hamiltonian does not significantly increase the norm, then we obtain $\sqrt\sparseness$ scaling.

\subsection{Norms of components}
\label{sec:norms}

Now we present numerical results suggesting that for typical matrices, the spectral norms of the components are small, and $\brk(H)$ can be upper bounded by a constant. If we consider general Hamiltonians, then the norms of the components are almost always smaller than the norm of the original Hamiltonian.
(In this subsection, we use ``norm'' to mean the spectral norm.)
We tested general Hamiltonians by generating random Hermitian matrices with normally distributed elements. In no case was $\brk(H)$ more than $1.2$, as shown in Fig.\ \ref{fig:rand}. For large dimension, $\brk(H)$ approached 1.

We expect larger norms for the components when breaking up a Hamiltonian derived from a unitary matrix as discussed in Sec.\ \ref{sec:unit} below (see Eq.\ \eqref{eq:uh}). 
This is because unitaries can have a large difference between the spectral norm and the 1-norm, and the spectral norms of the individual components are bounded by the 1-norm of $H$.
To test this class of Hamiltonians, we generated random unitaries according to the Haar measure.
The values of $\brk(H)$ were larger than for random Hamiltonians, but were still no larger than $1.5$, as also shown in Fig.\ \ref{fig:rand}.

\begin{figure}
\centering
\includegraphics[width=0.45\textwidth]{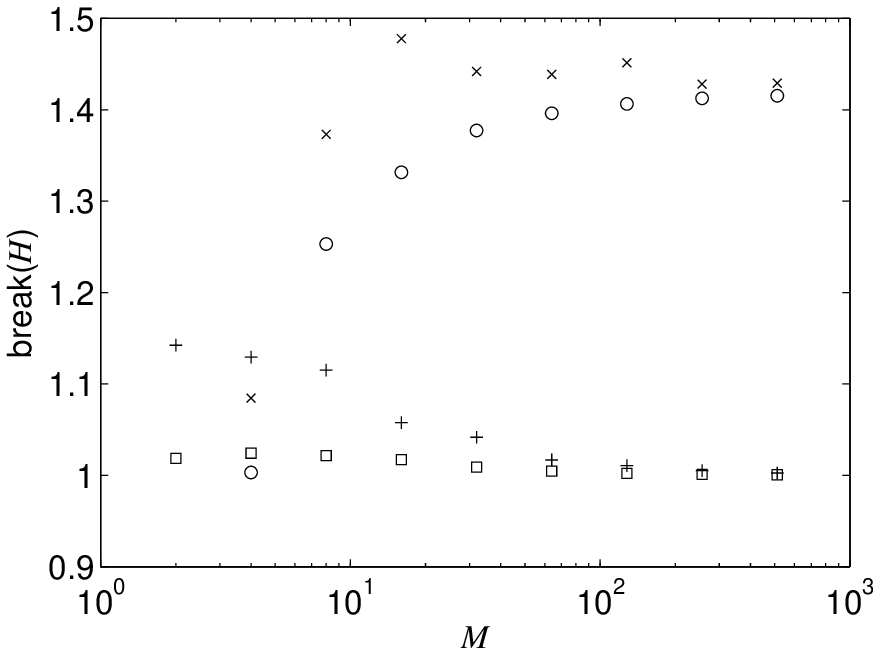}
\caption{The function $\brk(H)$ for random Hamiltonians. The plusses and squares are the maximum and mean values, respectively, obtained for 100 randomly generated Hermitian matrices. The crosses and circles are the maximum and mean values, respectively, for sets of 100 Hamiltonians composed of random unitaries.}
\label{fig:rand}
\end{figure}

One way to generate matrices that do have components with large norms is to perturb the quantum Fourier transform. We considered increasing the magnitude of the elements with positive real part by $0.01\%$ and decreasing the rest by $0.01\%$. The value of $\brk(H)$ (with $H$ constructed from this matrix as in Eq.\ \eqref{eq:uh}) was then proportional to $\sqrt{\HN}$ (see Fig.\ \ref{fig:path}).

\begin{figure}
\centering
\includegraphics[width=0.45\textwidth]{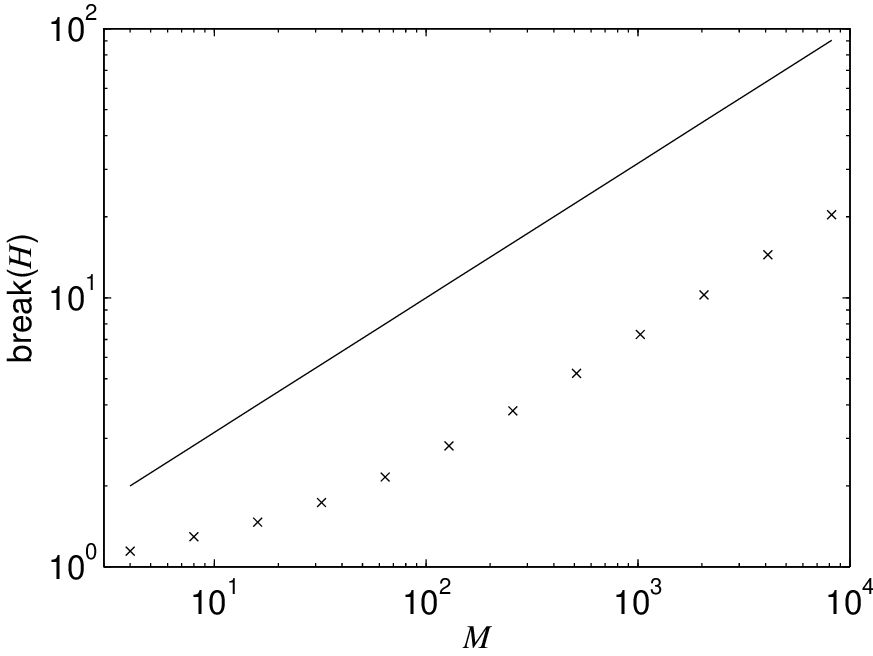}
\caption{The function $\brk(H)$ for a Hamiltonian composed of a matrix that has been produced by perturbing a quantum Fourier transform. The solid line is $\sqrt{M}$ for comparison.}
\label{fig:path}
\end{figure}

Although this example yields large norms for one splitting, the norms of the components are well-behaved with respect to other splittings. For example, a different threshold could be used, or we could introduce a smooth transition between the two components (i.e., for values in some transition region, part of the matrix element could go to one component and part to the other). However, we suspect that for any particular splitting, one can find examples that result in large norms for that splitting. 

\subsection{Large norms}
\label{sec:largenorms}

In this subsection we establish the improved simulation described in Theorem \ref{theorem2}, without relying on the assumption that the spectral norms of components remain small. We begin with an intuitive description of the method before giving the proof.

We again break the Hamiltonian into components $H_{\hj}$ according to the magnitudes of the matrix elements. In this case, the best available upper bound on the spectral norms of the components is $\La^2/A_\hj$. Using Lemma \ref{lemma6}, the number of queries to simulate each component involves a ratio $A_{\hj-1}/A_\hj^2$, in contrast to the corresponding ratio $A_{\hj-1}/A_\hj$ when the norms are assumed to be small (compare Eqs.\ \eqref{eq:witham} and \eqref{eq:rat2}). As a result, the cutoff values should be chosen to make $A_{\hj-1}/A_\hj^2$ constant, rather than to make $A_{\hj-1}/A_\hj$ constant.

In addition, the simulation of $H_{\hm}$ should not be performed via Lemma \ref{lemma6}, because that would result in an overall scaling no better than that provided by Lemma \ref{lemma6}. Instead, we use Theorem \ref{theorem1} to simulate $H_{\hm}$. By comparing the number of queries required to simulate $H_{\hm}$ and $H_{\hm-1}$, this means that (ignoring scaling in quantities other than $\sparseness$) we should have
\begin{equation}
\sqrt{\frac{\sparseness A_{\hm-2}}{A_{\hm-1}^2}} \approx \sparseness A_{\hm-1}.
\end{equation}
Here the expression on the left comes from using Lemma \ref{lemma6} to simulate $H_{\hm-1}$, and the expression on the right comes from using Theorem \ref{theorem1} for $H_{\hm}$. This expression means that $A_{\hm-2} \approx \sparseness A_{\hm-1}^4$. The restriction $A_{\hm-2}>A_{\hm-1}$ then means that $A_{\hm-1}\gtrapprox \sparseness^{-1/3}$. Therefore the number of queries is minimi{\s}ed for $A_{\hm-1} \approx A_{\hm-2} \approx \sparseness^{-1/3}$. Then $A_{\hj-1}/A_\hj^2 \approx \sparseness^{1/3}$, and the number of queries is roughly $\sparseness^{2/3}$.

To ensure that $A_0$ is independent of $\sparseness$, we must modify the above choices slightly. We choose a small constant $\splitsmall$, and take $A_{\hm-1}\propto \sparseness^{\splitsmall-1/3}$ and $A_{\hj-1}/A_\hj^2 \propto \sparseness^{1/3+2\splitsmall}$. Iterating gives $A_{\hm-2}\propto \sparseness^{4\splitsmall-1/3}$, $A_{\hm-3}\propto \sparseness^{10\splitsmall-1/3}$, and so forth. The sequence needs to give $A_0$ independent of $\sparseness$, but because the coefficient of $\splitsmall$ increases exponentially, $\hm$ need vary only logarithmically in $\splitsmall$.

This approach results in a number of queries to simulate each $H_\hj$ proportional to $\sparseness^{2/3+\xi}$. By choosing $\splitsmall\propto 1/\log\sparseness$, $\sparseness^{\xi}$ is $O(1)$. In addition, $\hm$ varies doubly logarithmically in $\sparseness$, which gives a double-logarithmic factor in the overall scaling in Theorem \ref{theorem2}. In the proof below, scaling in all quantities is considered, so it is convenient to define a quantity $\short$ that includes $\sparseness$ together with the other quantities we have omitted here. The scaling is then given in terms of $\short$, rather than explicitly in terms of $\sparseness$.

In order to show the result rigorously, we need to carefully choose the time intervals, because these are lower bounded by the conditions \eqref{eq:rests} and \eqref{eq:rests2} of Lemma \ref{lemma6}, and upper bounded by the need to ensure that the error in the Trotter formula is sufficiently small. This is challenging, because the bounds for the different components differ significantly.  The bounds on the Trotter error for $H_\hj$ decrease with $\hj$, so the lower bound on the time interval for $H_1$ is greater than the upper bound on the time interval for $H_\hm$. Thus it is not possible to combine these elements in the same Trotter formula while adequately bounding the error. To overcome this problem, we use nested Trotter formulae. We use a higher-order Lie-Trotter-Suzuki formula for $H_2$ through $H_\hm$, in order to obtain sufficiently small error despite the large upper bound on the norm of $H_\hm$. Then we use the Strang splitting to combine this product formula with $H_1$.

\begin{proof}[Proof of Theorem \ref{theorem2}.]
The Hamiltonian $H$ is again broken into $\hm$ pieces as in Eq.\ \eqref{eq:split}, again with $H_\hj = H^{A_\hj A_{\hj-1}}$ for $\hj < \hm$ and $H_\hm = H^{0A_{\hm-1}}$. For $\hj<\hm$, the norms are upper bounded as
\begin{equation}
\|H_\hj\| \le \|H_\hj\|_1 \le \Ha^2/A_\hj.
\end{equation}
The spectral norm of $H_\hm$ can be bounded more strongly:
\begin{align}
\|H_\hm\| &= \left\| H-\sum_{\hj=1}^{\hm-1} H_\hj\right\| \nonumber \\
&\le \Ha+\Ha^2/A_{\hm-1}.
\end{align}
Therefore, we can take $\La^{(\hj)}=\La^2/A_\hj$ for $\hj<\hm$ and $\La^{(\hm)}=\La+\La^2/A_{\hm-1}$. We can also take $\Li^{(\hj)}=\La^2/A_\hj$ for $\hj<\hm$, but for $\hj=\hm$ the best available bound gives $\Li^{(\hm)}=\La\sqrt{\sparseness}$. We have $\Lm^{(\hj)}=A_{\hj-1}$.

For $k\ge 1$, let
\begin{equation}
A_{\hm-k}=\La/\short^{1/3-(3\times 2^{k-1}-2)\splitsmall}
\end{equation}
where
\begin{align}
\short &\defeq \frac{\error \sparseness}{\hm\La t}, \\
\label{eq:splitsmall}
\splitsmall &\defeq \frac{1}{6(3\times 2^{\hm-2}-1)}.
\end{align}
With this choice, $A_0 = \La$ and $A_{\hm-1} = \La/\short^{1/3-\splitsmall}$; unlike in Section \ref{sec:smallnorms}, the ratio between successive cutoffs is not constant.
We break $H$ into
\begin{equation}
\label{eq:lval}
\hm \defeq \left\lceil \log_2 \left[ \frac 29 \log\left( \frac{\error \sparseness}{\La t} \right) + \frac 43 \right] \right\rceil
\end{equation}
pieces.

Note that if $\error \sparseness/\La t\le e^3$, then $L=1$, and we do not break up the Hamiltonian; we simply simulate $H$ using Theorem \ref{theorem1}.  Recall that by assumption, $\error \sparseness > \La t > \sqrt\error$.  Therefore, $\sparseness \La t > \La t/\sqrt\error > 1$. Since $\La \ge \|H\|_{\max}$, this shows that Theorem \ref{theorem1} uses $O(\sparseness \La t)$ queries.  Assuming $\error \sparseness/\La t\le e^3$, we have
\begin{align}
\sparseness \La t
&\le \sparseness \La t \left(\frac{e^3 \La t}{\sparseness\error}\right)^{1/3} \nn
&= e \sparseness^{2/3}\frac{(\La t)^{4/3}}{\error^{1/3}} \nn
&= O(\sparseness^{2/3}[(\log\log\sparseness)\La t]^{4/3} \error^{-1/3}).
\end{align}
This establishes Theorem \ref{theorem2} when $\error \sparseness/\La t\le e^3$. In the remainder of the proof, we assume that 
$\error \sparseness/\La t> e^3$, so $\hm>1$. It can also be shown that this implies $\short > 1$.

For $\hj<\hm$, Lemma \ref{lemma6} lets us simulate the Hamiltonian $H_\hj$ for time $\tau_\hj$ using
\begin{equation}
O\left( \tau_\hj^{3/2} \sqrt{\frac{\Lm^{(\hj)}\sparseness\Li^{(\hj)}\La^{(\hj)}}{\error_\hj}} \right)
\end{equation}
queries.
Conditions \eqref{eq:rests} and \eqref{eq:rests2} of Lemma \ref{lemma6} are satisfied provided $\tau_\hj$ is sufficiently small; we verify this below when choosing $\tau_\hj$ in the analysis of the Trotter error. The condition \eqref{eq:helpful} is trivial for $\hj<\hm$. We set $\error_\hj=\error\tau_\hj/\hm t$ to ensure that the contribution to the error from the simulations is $O(\error)$. Thus for $\hj<\hm$, the number of queries used to simulate $H_\hj$ for time $\tau_\hj$ is
\begin{equation}
\label{eq:rat2}
O\left(\La^2 \tau_\hj\sqrt{\frac{\hm\sparseness A_{\hj-1}t}{\error A_\hj^2}}\right).
\end{equation}
A simple calculation shows that
\begin{equation}
A_{\hj-1}/A_\hj^2 = \short^{1/3+2\splitsmall}/\La.
\label{eq:ratiowithsquare}
\end{equation}
Thus the query complexity of simulating $H_{\hj}$ for $\hj<\hm$ is
\begin{equation}
\label{eq:bbsc}
O\left(\sparseness^{2/3}\La\tau_\hj\short^\splitsmall\left(\frac{\hm \La t}{\error}\right)^{1/3}\right).
\end{equation}

To simulate $H_\hm$ for time $\tau_\hm$, we apply Theorem \ref{theorem1}, at a cost of
\begin{equation}
\label{eq:th1ex}
O\left( \frac{\La^2\sqrt{\hm t\tau_\hm}}{A_{\hm-1}\sqrt{\error}} + \sparseness A_{\hm-1}\tau_\hm + 1\right)
\end{equation}
queries.
By a simple calculation,
\begin{equation}
\sparseness A_{\hm-1} \tau_\hm = \sparseness^{2/3} \La \tau_\hm \short^\splitsmall\left(\frac{\hm \La t}{\error}\right)^{1/3},
\end{equation}
so the query complexity of simulating $H_\hm$ is also given by \eqref{eq:bbsc} provided the second term of Eq.\ \eqref{eq:th1ex} is dominant.
We verify this after choosing $\tau_\hm$ below.

Now we analyze the Trotter error.  We use a two-step process to combine the terms of $H$. First we use a Trotter formula for the two components $H_1$ and $\sum_{\hj=2}^{\hm}H_{\hj}$. Then we combine the terms of $\sum_{\hj=2}^{\hm}H_{\hj}$ using another Trotter formula. We do this because large time steps are needed for $H_1$, but its norm is small, whereas the time steps for the remaining $H_{\hj}$ can be smaller, but the norms are larger.

To combine $H_1$ and $\sum_{\hj=2}^{\hm}H_{\hj}$, the minimum time step is set by the restrictions \eqref{eq:rests} ($\La^{(\hj)}\tau_\hj \ge \sqrt{\error_\hj}$) and \eqref{eq:rests2} ($\tau_\hj \ge \La^{(\hj)}/\Lm^{(\hj)}\Li^{(\hj)}\sparseness$) for $H_1$. For general $\hj$, using the choice $\error_\hj = \error \tau_\hj/\hm t$, we see that these restrictions are satisfied provided
\begin{equation}
\tau_\hj \ge \max \left\{ \frac{\error}{\hm(\La^{(\hj)})^2t}, \frac{\La^{(\hj)}}{\Lm^{(\hj)}\Li^{(\hj)}\sparseness}\right\}.
\label{eq:outertimestep}
\end{equation}
For $\hj<\hm$, a simple calculation shows that
\begin{align}
\frac{\La^{(\hj)}}{\Lm^{(\hj)}\Li^{(\hj)}\sparseness}
&= \frac{1}{A_{\hj-1}\sparseness} \nn
&= \frac{\error}{\hm \La^2 t} \short^{-2/3-(3\times2^{\hm-\hj}-2)\splitsmall} \nn
&= \frac{\error}{\hm (\La^{(\hj)})^2 t} \short^{-6(2^{\hm-\hj}-1)\splitsmall},
\end{align}
where in the third line we have used
\begin{equation}
\frac{\La^{(\hj)}}{\La} = \frac{\La}{A_\hj}
= \short^{1/3 - (3 \times 2^{\hm-\hj-1}-2)\splitsmall}.
\label{eq:usefulratio}
\end{equation}
Therefore, since $\short > 1$, the first term of Eq.\ \eqref{eq:outertimestep} is larger than the second, and it suffices to take
\begin{equation}
\tau_\hj \ge \frac{\error}{\hm(\La^{(\hj)})^2t}
= \frac{\error A_\hj^2}{\hm\La^4t}.
\label{eq:tausufficient}
\end{equation}

Since $\short > 1$ implies $A_\hj<A_{\hj-1}$, this lower bound decreases with increasing $\hj$. Thus it suffices to ensure that $\tau_1$ is sufficiently large.
Here we use the $K=1$ integrator, so the ratio $t/\tau_1$ must be an even integer. We can achieve this, and ensure $\tau_1\ge{\error A_1^2}/{\hm\La^4 t}$, by taking
\begin{equation}
\tau_1 = \frac{t}{2\left\lfloor \frac {t^2\hm\La^4}{2\error A_1^2}\right\rfloor}.
\end{equation}
This expression is finite because
\begin{equation}
\frac {t^2\hm\La^4}{2\error A_1^2} = \frac {(\La t)^2\hm\short^{1/3+2\splitsmall}}{2\error}>\frac {\hm\short^{1/3+2\splitsmall}}{2} > 1.
\end{equation}
The equality uses Eq.\ \eqref{eq:ratiowithsquare} to compute $A_1$ in terms of $A_0=\La$, the first inequality uses the assumption $\La t > \sqrt\error$, and the last inequality uses $\short>1$ and $\hm \ge 2$.

The norms of the two components in the Trotter formula are bounded as
\begin{align}
\|H_1\| &\le \Ha^2/A_1 \le \La^2/A_1, \\
\left\| \sum_{\hj=2}^{\hm}H_{\hj} \right\| &\le \Ha+\Ha^2/A_1 \le 2 \La^2/A_1,
\end{align}
where the second line uses $A_1 \le A_0 = \La$.
Thus, by Eq.\ \eqref{eq:inter}, the Trotter error for combining $H_1$ and $\sum_{\hj=2}^\hm H_\hj$ with a $K=1$ integrator is
\begin{align}
O\left( \frac{\tau_1^2\La^6 t}{A_1^3}\right)
&= O\left( \frac{\error^2 A_1}{\hm^2\La^2t}\right) \nn
&\le O\left( \frac{\error^2 }{\hm^2 \La t}\right) \nn
&\le O(\delta),
\end{align}
where in the last step we have used $\error<\La t$, which follows from $\sqrt\error<\La t$ and $\error\le 1$.

Next we combine the $H_{\hj}$ with $\hj>1$, giving a simulation for time $\tau_1$.  We assume that $\hm\ge 3$ so there are at least two such terms to combine; then $\error\sparseness/\La t>e^{12}$.  By Eq.\ \eqref{eq:tausufficient} for $\hj=2$, the conditions of Lemma \ref{lemma6} are satisfied if we use time intervals of at least ${\error A_2^2}/{\hm\La^4 t}$.  However, we must choose time intervals that are compatible with the form of the integrator.  In this case, we use the $K=2$ integrator (see Section~\ref{sec:smallnorms}), which involves using two different intervals denoted $\tau_2^{(1)}$ and $\tau_2^{(2)}$.  We use the same pair of intervals for all $\hj\ge 2$.

The integrator requires intervals of the form 
\begin{align}
 \tau_2^{(1)} &= p_2 \tau_1/2\nu \\
 \tau_2^{(2)} &= (4p_2-1)\tau_1/2\nu
\end{align}
for some positive integer $\nu$, where $p_2 \defeq 1/(4-4^{1/3})$. Since $p_2<4p_2-1$, $\tau_2^{(2)} > \tau_2^{(1)}$, so to satisfy the conditions of Lemma \ref{lemma6}, it suffices to ensure that $\tau_2^{(1)}\ge {\error A_2^2}/{\hm\La^4 t}$.  We enforce this by choosing
\begin{align}
 \nu \defeq \left\lfloor \frac{p_2\tau_1 \hm \La^4 t}{2\error A_2^2} \right\rfloor .
\end{align}
This is a positive integer because
\begin{equation}
\frac{p_2\tau_1 \hm \La^4 t}{2\error A_2^2} \ge \frac{p_2 A_1^2}{2 A_2^2} = \frac{\short^{1/6+\splitsmall}}{2(4-4^{1/3})} >1.
\end{equation}
The first inequality uses $\tau_1\ge{\error A_1^2}/{\hm\La^4 t}$. The final inequality holds since $\error\sparseness/\La t>e^{12}$, as discussed above.

With this choice in hand, we can now verify that the second term of Eq.\ \eqref{eq:th1ex} is dominant. Henceforth we omit the superscripts on the time intervals, as they only differ by a multiplicative constant.
Since $A_2 = \La/\short^{1/4+3\splitsmall/2}$, $\tau_\hm = \tau_2 = \Theta(\short^{1/2-3\splitsmall}/\sparseness\La)$, and the second term of Eq.\ \eqref{eq:th1ex} is
\begin{align}
\sparseness A_{\hm-1} \tau_{\hm}
&= \Theta(A_{\hm-1}\short^{1/2-3\splitsmall}/\La)
= \Theta(\short^{1/6-2\splitsmall}).
\label{eq:secondterm}
\end{align}
In comparison, the first term is
\begin{align}
\frac{\La^2 \sqrt{\hm t \tau_\hm}}{A_{\hm-1}\sqrt{\delta}}
&= \Theta\left(\La \short^{1/3-\splitsmall} \frac{1}{\sqrt{\La \short}} \sqrt{\frac{\short^{1/2-3\splitsmall}}{\La}}\right) \nn
&= \Theta(\short^{1/12-5\splitsmall/2}),
\end{align}
which is smaller than \eqref{eq:secondterm} since $\short>1$.
We claim that the third term of Eq.\ \eqref{eq:th1ex} can also be neglected.
To see this, first note that the choice of $\hm$ in Eq.\ \eqref{eq:lval} ensures that $\splitsmall\le 1/\log\short$, so $\short^\splitsmall \le e$. By Eq.\ \eqref{eq:secondterm}, this implies that $\sparseness A_{\hm-1} \tau_{\hm} = \Theta(\short^{1/6-2\splitsmall}) = \Omega(1)$.
It follows that Eq.\ \eqref{eq:bbsc} also gives an upper bound on the number of queries needed to simulate $H_\hm$ for time $\tau_\hm$.

Now we analyze the error in the Trotter formula for $\sum_{\hj=2}^\hm H_\hj$.
The norm of the $H_\hj$ for $\hj \ge 2$ is largest for $\hj = \hm$, in which case we have the bound
\begin{align}
\|H_\hm\| &\le \Ha+\Ha^2/A_{\hm-1} \nonumber \\
&= O\left(\La\short^{1/3-\splitsmall}\right).
\end{align}
By Eq.\ \eqref{eq:inter}, the error in the $K=2$ integrator is
\begin{align}
O\left( \left[ \frac{\hm \tau_2 \La^2}{A_{\hm-1}}\right]^{5} \frac{t}{\tau_2}\right)
&= O\left( \frac{L \error^4 A_2^8 \short^{5/3 - 5\splitsmall}}{\La^{11}t^3} \right) \nn
&= O\left( \frac{\error^4 \hm }{(\La t)^3 \short^{1/3+17\splitsmall}}\right) \nn
&\le O\left( \error \left[\frac{\error}{\La t}\right]^{8/3}\frac{\hm^{4/3}}{\sparseness^{1/3}}\right) \nn
&\le O(\delta).
\end{align}
In the last line we have used the assumption $\error<\La t$ and Eq.\ \eqref{eq:lval} for $\hm$, which shows that $\hm=O(\log\log\sparseness)$.

So far we have only considered the number of queries to simulate the individual $H_{\hj}$. For the complete simulation, there is an additional factor of $\hm$ to take account of the integrators. Therefore, the total number of queries is
\begin{equation}
\label{eq:splitsca}
O\left( \frac{\sparseness^{2/3}(\hm\La t)^{4/3}\short^\splitsmall}{\error^{1/3}} \right).
\end{equation}
As discussed above, $\short^\splitsmall\le e$, so this factor can be ignored.  Overall, we find that
\begin{equation}
O\left( \frac{\sparseness^{2/3} {[(\log\log\sparseness)\La t]^{4/3}}}{\error^{1/3}} \right)
\end{equation}
queries suffice for the simulation, as claimed.
\end{proof}

\section{Implementation of unitaries}
\label{sec:unit}

Next we explain how to implement a unitary transformation using the results for simulation of Hamiltonians.
A simple way to implement a unitary transformation $U$, as proposed by Jordan and Wocjan \cite{Jordan09} (and independently observed by one of us), is to simulate the Hamiltonian
\begin{equation}
\label{eq:uh}
H=\begin{bmatrix}
0 & U  \\
U^\dagger & 0  \\
\end{bmatrix}.
\end{equation}
The Hilbert space consists of a qubit tensored with the target space.
Since $H^2=\openone$, we have
\begin{equation}
e^{-iHt} = \cos(t) \openone - i \sin(t) H,
\end{equation}
and applying this Hamiltonian for time $t=\pi/2$ yields the evolution
\begin{equation}
\label{eq:pi2ev}
e^{-iH\pi/2} \ket{1}\ket{\psib} = -i \ket{0}U\ket{\psib},
\end{equation}
which is sufficient to implement $U$.

Properties of the unitary $U$ and its associated Hamiltonian in Eq.\ \eqref{eq:uh} are closely related. The dimension of the Hamiltonian, $\HN$, is simply twice the dimension of the unitary, $\UN$.  In addition, we have
\begin{align}
\label{eq:relnorm1}
\Ha&=\|U\|=1, \\
\Hi&=\max\{\|U\|_1,\|U^\dagger\|_1\}, \\
\label{eq:relnorm3}
\Hm&=\|U\|_{\max}.
\end{align}

We assume that the matrix elements of $U$ are given by an oracle $O_U$ as in Eq.\ \eqref{eq:uoracle}. This oracle can trivially be used to construct an oracle $O_H$ for the Hamiltonian as in Eq.\ \eqref{eq:horacle}. Each call to $O_H$ uses one call to $O_U$, so black-box Hamiltonian simulation results can be applied directly to black-box unitary implementation. However, for unitary implementation we can take advantage of the fact that the eigenvalues of the Hamiltonian are restricted.
\begin{lemma}
\label{lemma3}
Suppose $H$ has eigenvalues $\pm 1$ and $\pi/[2\arcsin(\unlaziness/\Li)]$ is an odd integer, for $\unlaziness\in(0,1]$.  Using a quantum walk with states $\ket{\phib_j}$ as in Eq.\ \eqref{eq:alt3}, evolution for time $\pi/2$ can be simulated exactly using $O(\Li/\unlaziness)$ queries.
\end{lemma}
\begin{proof}
Since the eigenvalues of $H$ are $\lambda = \pm 1$, the relationship between $\arcsin\tilde\lambda$ and $\tilde\lambda$ is simple: taking
\begin{equation}
d = \frac{\pi}{2\arcsin(\unlaziness/\Li)},
\end{equation}
the eigenvalues of $V^d$ are $+i$ for $\lambda=1$ and $-i$ for $\lambda=-1$.  These eigenvalues are equivalent (up to the minus sign) to evolution under the Hamiltonian $H$ for time $\pi/2$.
\end{proof}

This result can be used to exactly implement unitary operators via a quantum walk. The scaling is as follows:
\begin{theorem}
\label{thm:exact}
Given a black-box unitary $U$, let $\Lm\ge\|U\|_{\max}$.
Then $U$ can be implemented exactly with $O\left( \UN \Lm \right)$ queries to $O_U$.
\end{theorem}
Since we are primarily concerned with implementation of general unitaries, which are not sparse, we express unitary implementation results in terms of the dimension $\UN$ rather than the sparseness parameter $\sparseness$.
\begin{proof}
This implementation proceeds by simulating the Hamiltonian given in Eq.\ \eqref{eq:uh} for time $\pi/2$ using Lemma \ref{lemma3}, with the steps of the quantum walk implemented using Lemma \ref{lemma1}. The Hamiltonian has no more than $\UN$ nonzero elements in any row of column, so we can take $\sparseness=\UN$.

Take $\unlaziness=\Li/\UN X$, where
\begin{align}
\label{eq:X}
X &= \frac 1 {\UN\sin[\pi/(2d)]}, \\
\label{eq:d}
d &= 2\left\lceil \frac{\pi}{4\arcsin[1/(\Lm\UN)]}-\frac 12 \right\rceil+1.
\end{align}
It is easily shown that $X\ge \Lm$, so $\unlaziness\le\Li/D\Lm\le 1$. In addition, $\pi/[2\arcsin(\unlaziness/\Li)]$ is an odd integer, so the conditions of Lemma \ref{lemma3} are satisfied. Then, using Lemma \ref{lemma3}, the Hamiltonian can be simulated for time $\pi/2$ using $O(\Li/\unlaziness)=O(\UN X)$ steps of the quantum walk.

Because $\unlaziness\le\Li/D\Lm$, we can use Lemma \ref{lemma1}, and each step of the quantum walk can be implemented using $O(1)$ queries. Thus the total number of queries is $O(\UN X)$. Because $\Lm\ge \Hm \ge 1/\sqrt{\UN}$, $X\le 2\Lm$, and the number of queries is $O(\UN\Lm)$.
\end{proof}

Our other results on Hamiltonian simulation can also be used to implement unitaries, although in these cases there are other sources of error, so the simulation can no longer be performed exactly.
In each case we take $t=\pi/2$, $\Ha=1$, and $\sparseness=\UN$.
Lemma \ref{lemma6} yields the following corollary for unitary implementation.
\begin{corollary}
\label{corr11}
Given a black-box unitary U, let $\Li\ge\max\{\|U\|_1,\|U^\dagger\|_1\}$ and $\Lm\ge\|U\|_{\rm max}$.  Then $U$ can be implemented with error at most $\error\in(0,1]$ using
\begin{equation}
\label{eq:cor11}
O\left( \sqrt{{\Lm\UN\Li}/{\error}}\right)
\end{equation}
queries to $O_U$.
\end{corollary}

\begin{proof}
We apply Lemma \ref{lemma6} together with the norm bounds in Eqs.\ \eqref{eq:relnorm1} to \eqref{eq:relnorm3}. We use $t=\pi/2$ and $\Ha=1$ to obtain Eq.\ \eqref{eq:cor11}. We omit conditions \eqref{eq:rests} to \eqref{eq:helpful} of Lemma \ref{lemma6} as they are automatically satisfied. First, the condition \eqref{eq:rests} holds because $\error\le 1$. Second, \eqref{eq:rests2} holds because $\Lm\ge 1/\sqrt{N}$ and $\Li\ge\|U\|_1\ge 1$. Third, \eqref{eq:helpful} holds because $\Li\ge 1=\La$.
\end{proof}

Similarly, Theorem \ref{theoremsm} yields the following.

\begin{corollary}
\label{corollsm}
Let $\sca\in[\brk(U),\sqrt{\UN}]$.  The unitary operation $U$ can be implemented with error at most $\error\in(0,1]$ using
\begin{equation}
\label{ex:smcor}
O\left( \sqrt{\sca\UN/\error} (\log\UN)^{7/4} \right)
\end{equation}
queries to $O_H$ and $O_\F$, provided $\delta\UN > \pi/2$.
\end{corollary}

\begin{proof}
We use Theorem \ref{theoremsm} with $\La=\Ha=1$, $t=\pi/2$, and $\sparseness=\UN$.
Then the number of queries is as in Eq.~\eqref{ex:smcor}.
Since $\Ha t=\pi/2$, the condition $\delta\sparseness > \La t > \sqrt\error$ in Theorem \ref{theoremsm} becomes $\delta\UN > \pi/2 > \sqrt\error$, and since $\error \le 1$, the latter inequality is trivial.
\end{proof}

Finally, Theorem \ref{theorem2} yields Corollary \ref{corollary1}, which can be proven as follows.

\begin{proof}[Proof of Corollary \ref{corollary1}.] For unitaries, $\La = \Ha=1$, $t=\pi/2$, and $\sparseness=\UN$, so Eq.\ \eqref{eq:th2} gives an upper bound of
\begin{equation}
O\left( \UN^{2/3} (\log\log\UN)^{4/3} \error^{-1/3} \right)
\end{equation}
queries, as claimed.
The condition $\delta\sparseness > \La t > \sqrt\error$ becomes $\delta\UN > \pi/2$ for the same reason as in the proof of Corollary \ref{corollsm} above.
If $\delta\UN \le \pi/2$, then we instead implement the unitary using Theorem \ref{thm:exact}.  This takes $O(N)$ queries, which is smaller than the claimed upper bound.
\end{proof}

In the worst case, Corollary \ref{corr11} yields query complexity of $O(\UN^{3/4}/\sqrt\error)$.
This is because $\Lm$ could be as large as 1 and $\Li$ could be as large as $\sqrt \UN$.
On the other hand, if the nonzero matrix elements are of similar magnitude, then $\Lm\propto 1/\Li$, so the scaling will be $O(\sqrt{\UN/\error})$.
Alternatively, if it is possible to break the unitary into components without otaining large spectral norms, then $\brk(U)=O(1)$, and Corollary \ref{corollsm} yields scaling of $\tilde O(\sqrt{\UN/\error})$.
Those results are not sufficient to prove this scaling for all unitaries, because $\brk(U)$ may be large.
However, in general we can use Corollary \ref{corollary1} to implement any unitary with $\tilde O( \UN^{2/3} \error^{-1/3} )$ queries.

\section{Examples}
\label{sec:examp}

We now consider some simple examples of unitaries and discuss the query complexity of implementing them by the methods of the previous section.

First, consider the unitary with matrix elements $U_{jk} = g(j+k \bmod \UN)$, where $g$ is a black-box function for a search problem with a unique marked item $j^\star$. The function $g\colon \{1,2,\ldots,N\} \to \{0,1\}$ satisfies $g(j^\star)=1$ and $g(j)=0$ for $j\ne j^\star$. Since $U\ket{0} = \ket{j^\star}$, implementing $U$ solves the search problem; thus it requires $\Omega(\sqrt{\UN})$ queries \cite{Grovopt}. This unitary has $\|U\|_{\rm max}=1$ and $\|U\|_1=1$, so Corollary \ref{corr11} gives a complexity of $O(\sqrt {\UN/\error})$, which is optimal. In fact, simply implementing the isometry $T$ solves the search problem, because it can prepare the state $T\ket{0} = \ket{0}\ket{j^\star}$. The implementation of $T$ in this case is in fact equivalent to the standard Grover search algorithm \cite{Grover96}.

In this case, $\Sj$ is known, so the implementation can be performed exactly. Using Lemma \ref{lemma3}, unitaries may be implemented exactly using a quantum walk, so the only remaining source of error is in performing the steps of the quantum walk. From the proof of Lemma \ref{lemma4}, the steps of the quantum walk may be performed exactly if $r_j^{\rm opt}$ from Eq.\ \eqref{eq:rjopt} is a known integer. Because we have $\Sj=\unlaziness=\Li=1$, we can easily adjust $X$ to ensure that this is the case, and therefore that the simulation is performed exactly. More generally, whenever $U$ is a permutation matrix, it can be implemented in only $O(\sqrt{\UN})$ queries in a similar fashion.

Another simple example is the quantum Fourier transform, the unitary with $U_{jk}=e^{2\pi ijk/\UN}/\sqrt{\UN}$. For this unitary, $\|U\|_{\rm max}=1/\sqrt{\UN}$ and $\|U\|_1=\sqrt{\UN}$. Therefore, Corollary \ref{corr11} again gives a complexity of $O(\sqrt {\UN/\error})$. In fact, for this case we can take $\unlaziness=1$, $\Sj=\Li=\sqrt{\UN}$ and $X=1/\sqrt{N}$, so Eq.\ \eqref{eq:rjopt} gives $r_j^{\rm opt}=0$. Therefore no amplitude amplification is required, and the implementation is again exact.

These two examples illustrate the two extremal cases where Corollary \ref{corr11} gives scaling of $\sqrt{\UN}$.
First, if all the weight is on one matrix element in each row, then $\|U\|_1=1$.
At the other extreme, if the weight is evenly distributed between the matrix elements, then $\|U\|_1=\sqrt{\UN}$, but $\|U\|_{\rm max}=1/\sqrt{\UN}$.
These correspond to the two points listed at the end of Sec.\ \ref{sec:full}.
In either case, the nonzero matrix elements have the same magnitude.

Note that to take advantage of sparsity, the locations of the nonzero elements must be known (or more precisely, their locations must be accessible via the oracle $O_F$).
Effectively, the quantity $\sparseness$ measures how many matrix elements are not known to be zero.
For the search problem, the locations of the nonzero elements are not known in advance (finding those positions would in itself solve the search problem), so $\sparseness=\UN$.
In contrast, for the norms, it does not matter if the nonzero elements are in known positions.
If there are $m$ nonzero matrix elements, then $\|U\|_1\le\sqrt{m}$ regardless of the positions of those elements.

For Corollary \ref{corr11} to yield scaling worse than $\sqrt{\UN}$, the distribution of magnitudes of matrix elements of $U$ must have a sharp peak in combination with a relatively broad distribution for the remaining elements.
As a natural example of this, consider the unitary given by $U=\exp(-i\pi J_x/2)$, where $J_x$ is the $x$-rotation operator for a spin-$J$ system, with dimension $\UN=2J+1$, and we use the basis of $J_z$ eigenstates.
The first column of $\exp(-i\pi J_x/2)$ has a relatively narrow peak, whereas for columns towards the middle the elements are more spread out (see Fig.\ \ref{fig:rot}).
The maximum element of $U$ has absolute value $\sqrt{(2\ceil{J})!}/{2^{\ceil{J}} \ceil{J}!}$, which is $O(J^{-1/4})$ by Stirling's formula.
Since $\|U\|_1 = O(\sqrt{J})$, Corollary \ref{corr11} yields an overall number of black-box queries of $O(\UN^{5/8}/\error^{1/2})$.

\begin{figure}
\centering
\includegraphics[width=0.45\textwidth]{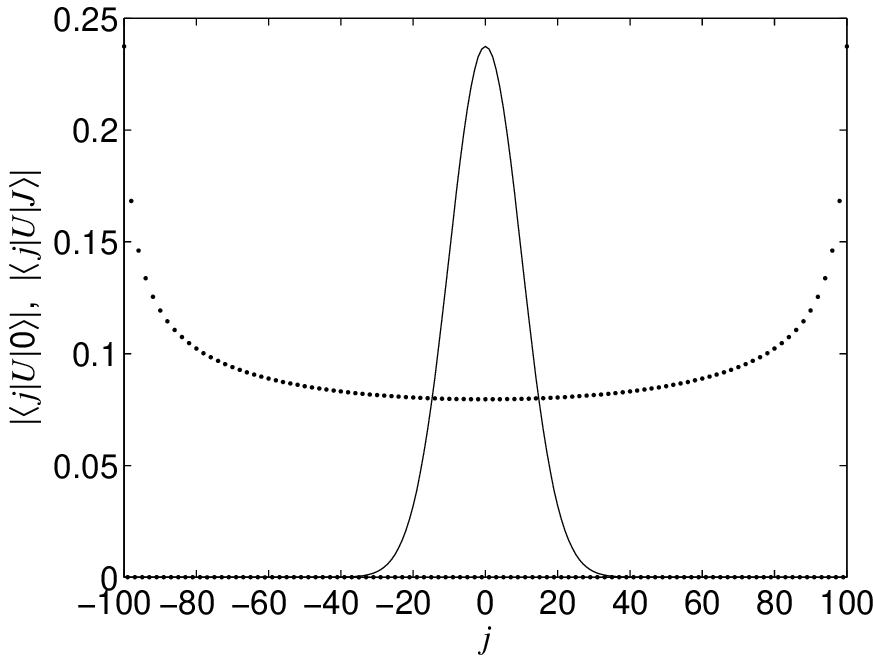}
\caption{The matrix elements of $U=\exp(-i\pi J_x/2)$ in the basis of $J_z$ eigenstates for $J=100$. The separate points are $|\bra{j}U\ket{0}|$, and the solid curve is $|\bra{j}U\ket{J}|$.}
\label{fig:rot}
\end{figure}

Thus Corollary \ref{corr11} does not provide $\sqrt{\UN}$ scaling in this case, though the scaling is better than the worst-case $\UN^{3/4}$ scaling.
Using Corollary \ref{corollary1} would not yield improved scaling in this example, because $2/3>5/8$.
However, numerical testing indicates that breaking this unitary into components does not increase the spectral norms, so using the approach given in Sec.\ \ref{sec:smallnorms} would yield $\sqrt{\UN}$ scaling.
In this example, calculating the matrix elements of $U$ is nontrivial, and consequently the overall complexity of the algorithm in terms of elementary gates would be greater than $\sqrt{\UN}$.

We emphasi{\s}e that the motivation to implement unitaries is not as a shortcut to simulation of Hamiltonians via $U=e^{-iHt}$. In general, calculating the matrix elements of $U=e^{-iHt}$ given the matrix elements of $H$ may be difficult. Rather, the motivation for implementing unitaries is to provide a tool to develop other algorithms. As discussed above, the search problem may be encoded as a unitary operation. The algorithm for implementing unitaries may be regarded as a new generali{\s}ation of the Grover algorithm.

\section{Conclusion}
\label{sec:conc}

We have shown how to use quantum walks to simulate black-box Hamiltonians.
In particular, we showed that these techniques can be used to implement an arbitrary $\UN\times\UN$ unitary transformation using $\tilde O(\UN^{2/3}/\error^{1/3})$ queries to a black box for its matrix elements, with error at most $\error$ as quantified by the trace distance.

Our approach is based on simulating Hamiltonian dynamics via discrete-time quantum walk \cite{Childs08}, combined with state preparation via amplitude amplification \cite{Grover00} and integrators to break up the Hamiltonian. In many cases the implementation can be performed even faster, with $\tilde O(\sqrt{\UN/\error})$ black-box calls. This scaling can be achieved except when breaking the Hamiltonian into a sum of terms yields components with large spectral norms, and numerical testing suggests that such cases are rare.

For many applications, our work provides the best known simulation of sparse Hamiltonians.  The number of queries is strictly linear in $\Ha t$, rather than slightly superlinear, as when higher-order integrators are used \cite{Berry07,Childs04}. In addition, the scaling in the sparseness parameter $\sparseness$ is at worst linear, in contrast with the $O(\sparseness^4)$ scaling of Ref.\ \cite{Berry07}.

The best lower bound we know for black-box unitary implementation is $\Omega(\sqrt{\UN})$ queries, because implementing an $\UN \times \UN$ unitary suffices to solve unstructured search with $\UN$ items. It remains an open problem to determine whether it is possible to perform the simulation using $O(\sqrt{\UN})$ queries in general.

Our results also apply to more general Hamiltonian simulation problems.  It might be interesting to investigate the extent to which the general black-box Hamiltonian simulation described by Theorem \ref{theorem2} can be improved.  Simulations using $O(\|Ht\|)$ queries are not possible in general \cite{Childs09}, but the tradeoff between quantities such as $\sparseness$, $\|Ht\|$, $\|Ht\|_1$, $\|Ht\|_{\max}$, and $\error$ is poorly understood.

\acknowledgements

We thank Richard Cleve, Aram Harrow, Stephen Jordan, Robin Kothari, and John Watrous for helpful discussions.
DB is funded by an Australian Research Council Future Fellowship (FT100100761).
AMC received support from MITACS, NSERC, QuantumWorks, and the US ARO/DTO.


\begin{thebibliography}{99}
\bibitem{Feynman82}
R. P. Feynman,
\emph{Simulating physics with computers},
International Journal of Theoretical Physics \textbf{21}, 467 (1982).
\bibitem{Childs03} 
A. M. Childs, R. Cleve, E. Deotto, E. Farhi, S. Gutmann, and D. A. Spielman,
\emph{Exponential algorithmic speedup by quantum walk},
in \emph{Proceedings of the 35th ACM Symposium on Theory of Computing} (ACM, New York, 2003), pp. 59--68; arXiv:quant-ph/0209131.
\bibitem{Farhi08}
E. Farhi, J. Goldstone, and S. Gutmann,
\emph{A quantum algorithm for the {H}amiltonian {NAND} tree},
Theory of Computing \textbf{4}, 169 (2008); arXiv:quant-ph/0702144.
\bibitem{Childs09b}
A. M. Childs, R. Cleve, S. P. Jordan, and D. Yonge-Mallo,
\emph{Discrete-query quantum algorithm for NAND trees},
Theory of Computing \textbf{5}, 119 (2009);
quant-ph/0702160.
\bibitem{Harrow09}
A. W. Harrow, A. Hassidim, and S. Lloyd,
\emph{Quantum algorithm for linear systems of equations},
Phys. Rev. Lett. \textbf{103}, 150502 (2009).
\bibitem{Lloyd96}
S. Lloyd, 
\emph{Universal quantum simulators},
Science \textbf{273} 1073 (1996).
\bibitem{Aharonov03}
D. Aharonov and A. Ta-Shma,
\emph{Adiabatic quantum state generation and statistical zero knowledge},
in \emph{Proceedings of the 35th ACM Symposium on Theory of Computing, 2003} (ACM, New York, 2003), pp. 20--29; arXiv:quant-ph/0301023.
\bibitem{Childs04} 
A. M. Childs,
\emph{Quantum information processing in continuous time},
Ph.D. thesis, Massachusetts Institute of Technology, 2004.
\bibitem{Berry07} 
D. W. Berry, G. Ahokas, R. Cleve, and B. C. Sanders, 
\emph{Efficient quantum algorithms for simulating sparse Hamiltonians},
Commun. Math. Phys. \textbf{270}, 359 (2007);
arXiv:quant-ph/0508139.
\bibitem{Childs08} 
A. M. Childs, 
\emph{On the relationship between continuous- and discrete-time quantum walk},
Commun. Math. Phys. \textbf{294}, 581 (2009);
arXiv:0810.0312.
\bibitem{Childs11}
A. M. Childs and R. Kothari,
\emph{Simulating sparse Hamiltonians with star decompositions},
Theory of Quantum Computation, Communication, and Cryptography (TQC 2010), Lecture Notes in Computer Science \textbf{6519}, 94 (2011);
arXiv:1003.3683.
\bibitem{Grover00} 
L. K. Grover,
\emph{Synthesis of quantum superpositions by quantum computation},
Phys. Rev. Lett. \textbf{85}, 1334 (2000).
\bibitem{Reck94}
M. Reck, A. Zeilinger, H. J. Bernstein, and P. Bertani,
\emph{Experimental realization of any discrete unitary operator},
Phys. Rev. Lett. {\bf 73}, 58 (1994).
\bibitem{Barenco95}
A. Barenco, C. H. Bennett, R. Cleve, D. P. DiVincenzo, N. Margolus, P. Shor,
T. Sleator, J. Smolin, and H. Weinfurter,
\emph{Elementary gates for quantum computation},
Phys. Rev. A {\bf 52}, 3457 (1995);
quant-ph/9503016.
\bibitem{SoloKit1}
A. Y. Kitaev,
\emph{Quantum computations: Algorithms and error correction},
Russ. Math. Surveys \textbf{52}, 1191 (1997).
\bibitem{SoloKit2}
A. Y. Kitaev, A. H. Shen, and M. N. Vyalyi,
\emph{Classical and Quantum Computation},
Graduate Studies in Mathematics Vol. 47 (American Mathematical Society, Providence, RI, 2002).
\bibitem{Harrow02}
A. W. Harrow, B. Recht, and I. L. Chuang, 
\emph{Efficient discrete approximations of quantum gates},
J. Math. Phys. \textbf{43}, 4445 (2002).
\bibitem{Knill95}
E. Knill,
\emph{Approximation by quantum circuits},
Technical Report LAUR-95-2225, Los Alamos National Laboratory, 1995; arXiv:quant-ph/9508006.
\bibitem{Jordan09} 
S. P. Jordan and P. Wocjan, 
\emph{Efficient quantum circuits for arbitrary sparse unitaries},
Phys. Rev. A {\bf 80}, 062301 (2009);
arXiv:0904.2211.
\bibitem{Grovopt} 
C. H. Bennett, E. Bernstein, G. Brassard, and U. Vazirani,
\emph{Strengths and weaknesses of quantum computing},
SIAM J. Comput. \textbf{26}, 1510 (1997);
arXiv:quant-ph/9701001.
\bibitem{Szegedy} 
M. Szegedy,
\emph{Quantum speed-up of Markov chain based algorithms},
Proceedings of the 45th IEEE Symposium on Foundations of Computer Science (IEEE, Los Alamitos, CA, 2004), pp. 32--41; arXiv:quant-ph/0401053.
\bibitem{Brassard97} 
G. Brassard and P. H{\o}yer, 
\emph{An exact quantum polynomial-time algorithm for Simon's problem},
in \emph{Proceedings of Fifth Israeli Symposium on Theory of Computing and Systems} (IEEE, Los Alamitos, CA, 1997), pp. 12--23; arXiv:quant-ph/9704027.
\bibitem{Brassard} 
G. Brassard, P. H{\o}yer, M. Mosca, and A. Tapp, in \emph{Quantum Computation and Information}, edited by S. J. Lomonaco and
H. E. Brandt (AMS, Providence, 2002); arXiv:quant-ph/0005055.
\bibitem{Grover98} 
L. K. Grover,
\emph{Quantum computers can search rapidly by using almost any transformation},
Phys. Rev. Lett. \textbf{80}, 4329 (1998).
\bibitem{Grover96}
L. K. Grover, 
\emph{A fast quantum mechanical algorithm for database search},
in \emph{Proceedings of the 28th Annual ACM Symposium on Theory of Computing} (ACM, New York, 1996), pp. 212--219;
arXiv:quant-ph/9605043.
\bibitem{Strang}
G. Strang,
\emph{On the construction and comparison of difference schemes},
SIAM J. Numer. Anal. \textbf{5}, 506 (1968).
\bibitem{Suzuki90}
M. Suzuki,
\emph{Fractal decomposition of exponential operators with applications to many-body theories and Monte Carlo simulations},
Phys. Lett. A \textbf{146}, 319 (1990).
\bibitem{Suzuki91}
M. Suzuki,
\emph{General theory of fractal path integrals with applications to many-body theories and statistical physics},
J. Math. Phys. \textbf{32}, 400 (1991).
\bibitem{Nathan10}
N. Wiebe, D. W. Berry, P. H{\o}yer, and B. C. Sanders,
\emph{Higher order decompositions of ordered operator exponentials},
J. Phys. A: Math. Theor. \textbf{43}, 065203 (2010);
arXiv:0812.0562.
\bibitem{Childs09}
A. M. Childs and R. Kothari, 
\emph{Limitations on the simulation of non-sparse Hamiltonians},
Quantum Inform. Comput. \textbf{10}, 669 (2010); arXiv:0908.4398.
\end{thebibliography}
\end{document}